\documentclass{jpp}

\usepackage{graphicx}
\usepackage{natbib}
\usepackage{xcolor}
\usepackage{amsmath}
\usepackage{accents}
\usepackage{esint}
\usepackage{bm}

\ifCUPmtlplainloaded \else
  \checkfont{eurm10}
  \iffontfound
    \IfFileExists{upmath.sty}
      {\typeout{^^JFound AMS Euler Roman fonts on the system,
                   using the 'upmath' package.^^J}%
       \usepackage{upmath}}
      {\typeout{^^JFound AMS Euler Roman fonts on the system, but you
                   dont seem to have the}%
       \typeout{'upmath' package installed. JPP.cls can take advantage
                 of these fonts, if you use 'upmath' package.^^J}%
      }
  \else
  \fi
\fi

\ifCUPmtlplainloaded \else
  \checkfont{msam10}
  \iffontfound
    \IfFileExists{amssymb.sty}
      {\typeout{^^JFound AMS Symbol fonts on the system, using the
                'amssymb' package.^^J}%
       \usepackage{amssymb}%
         
         \let\geq=\geqslant
      }{}
  \fi
\fi

\ifCUPmtlplainloaded \else
  \IfFileExists{amsbsy.sty}
    {\typeout{^^JFound the 'amsbsy' package on the system, using it.^^J}%
     \usepackage{amsbsy}}
    {}
\fi

\newsavebox{\astrutbox}
\sbox{\astrutbox}{\rule[-5pt]{0pt}{20pt}}

\newcommand\Lxio{\mathcal{L}_{\xi_0}}

\newtheorem{definition}{Definition}

\newtheorem{remark}{Remark}

\newtheorem{theorem}{Theorem}

\newtheorem{lemma}{Lemma}

\title[General formulas for adiabatic invariants]{General formulas for adiabatic invariants in nearly-periodic Hamiltonian systems}

\author[J. W. Burby]%
{J.\ns W.\ns B\ls U\ls R\ls B\ls Y$^1$ and\ns J.\ns S\ls Q\ls U\ls I\ls R\ls E$^2$}

\affiliation{$^1$Los Alamos National Laboratory, Los Alamos, New Mexico 87545, USA\\
$^2$Physics Department, University of Otago, Dunedin 9016, New Zealand}

\pubyear{2020}
\volume{\infty}
\pagerange{1--\infty}
\date{?; revised ?; accepted ?. - To be entered by editorial office}
\begin{document}

\maketitle

\begin{abstract}
While it is well-known that every nearly-periodic Hamiltonian system possesses an adiabatic invariant, extant methods for computing terms in the adiabatic invariant series are inefficient. The most popular method involves the heavy intermediate calculation of a non-unique near-identity coordinate transformation, even though the adiabatic invariant itself is a uniquely-defined scalar. A less well-known method, developed by S. Omohundro,  avoids calculating intermediate sequences of coordinate transformations but is also inefficient as it involves its own sequence of complex intermediate calculations.
In order to improve the efficiency of future calculations of adiabatic invariants, we derive generally-applicable, readily computable formulas for the first several terms in the adiabatic invariant series. To demonstrate the utility of these formulas, we apply them to charged particle dynamics in a strong magnetic field and magnetic field-line dynamics when the field lines are nearly closed.
\end{abstract}


\section{Introduction}

Adiabatic invariance historically played an essential role in the development of plasma physics, especially in the theory of charged particle motion in strong magnetic fields. See \cite{Cary_2009} for an in-depth review of the latter topic. While an adiabatic invariant is not a true conserved quantity, it is approximately conserved over large intervals of time, and is therefore just as good as a true invariant for many practical purposes. In this Article we will derive a new general formula for the adiabatic invariant associated with a \emph{nearly-periodic} Hamiltonian system. Such systems, along with their adiabatic invariants, were previously studied systematically in \cite{Kruskal_1962}.

Today the most popular method for computing adiabatic invariants involves near-identity coordinate transformations. First ``nice" coordinates are found in which the expression for the adiabatic invariant becomes simple. Then the inverse coordinate transformation is applied to find an expression for the adiabatic invariant in a simpler, more desirable coordinate system.  This approach is exemplified by Littlejohn's work on Hamiltonian formulations of guiding center dynamics in \cite{Littlejohn_1981}, \cite{Littlejohn_1982}, \cite{Littlejohn_1983}, and \cite{Littlejohn_1984}. Speaking more generally, at present there are (involved) \emph{procedures} for computing adiabatic invariants, but general-use \emph{formulas} for adiabatic invariants are unavailable.

The formula that we will obtain does not involve coordinate transformations. Instead it builds upon the coordinate-free ideas developed in \cite{Omohundro_1986} concerning the so-called \emph{roto-rate vector}. The roto-rate vector was first introduced in \cite{Kruskal_1962} as a vector field $\bm{R}$ that generates an approximate $U(1)$ symmetry for nearly-periodic systems. Kruskal recognized the physical and conceptual significance of the roto-rate vector, but did not know how to compute $\bm{R}$ without first introducing an infinite sequence of near-identity coordinate transformations. Over twenty years later, \cite{Omohundro_1986} showed that, in principle, $\bm{R}$ can be computed in any coordinate system without introducing near-identity coordinate transformations, and even gave an algorithm for carrying out the calculation order-by-order in perturbation theory. However, Omohundro's results stop short of providing general formulas for $\bm{R}$, presumably as a result of the cumbersome nature of his algorithm.

Our approach to deriving a general formula for a nearly-periodic Hamiltonian system's adiabatic invariant starts by improving Omohundro's algorithm for computing the roto-rate. The key to the improvement is recognizing that the messiest element of Omohundro's algorithm, namely enforcing that the integral curves of the roto-rate vector are $2\pi$-periodic, may be reimagined as a straightforward application of the famous Baker-Campbell-Hausdorff formula for the logarithm of composed exponentials. Using this improved algorithm we will push past Omohundro's results by deriving general-use, coordinate-independent formulas for the roto-rate. We will then feed these formulas into Noether's theorem for presymplectic Hamiltonian systems (see, e.g. \cite{Munteanu2014GeometricMF}) in order to identify coordinate-independent formulas for the adiabatic invariant.

Our principal motivation for deriving this new formula is a desire for computing adiabatic invariants in infinite-dimensional Hamiltonian systems. While coordinate transform methods (e.g. perturbative changes of dependent variables) can be applied to such systems, the complexity of the required calculations easily gets out of hand. Coordinate-independent formulas for a system's adiabatic invariant would bypass much of this tedium, and therefore comprise a more efficient route to the desired result.

That said, we will not present any infinite-dimensional example applications in this Article. Instead we will first use our new formula to reproduce the first two terms in the adiabatic invariant series for non-relativistic strongly magnetized charged particles. Then we will use our formula to calculate a coordinate-free expression for the field-line adiabatic invariant associated with a magnetic field whose lines of force are nearly closed. This adiabatic invariant defines approximate flux surfaces for this special class of magnetic fields, which includes near-axisymmetric-vacuum fields, and more generally any field that is close to an integrable field with constant rational rotational transform. It is worth remarking from the outset that this approximate flux function \emph{is not} provided by standard KAM theory, which crucially relies on unperturbed fields with non-vanishing shear.

As we derive the general formula we will make liberal use of the standard machinery for performing calculus on manifolds, which includes Lie derivatives, flows, pullbacks, differential forms, and Stoke's theorem. A complete and rigorous description of this machinery, along with a vast amount of useful information concerning the coordinate-independent approach to Hamiltonian systems, is given in \cite{Abraham_2008}. The recent tutorial \cite{MacKay_tutorial_2020} on differential forms for plasma physicists is also an invaluable resource. throughout the article we will adopt the notation $\fint Q(\theta)\,d\theta = (2\pi)^{-1}\int_0^{2\pi}Q(\theta)\,d\theta$ for averages over an angular variable $\theta\in U(1)$.

The systems that exhibit the adiabatic invariants we would like to compute have two essential features: (a) they are \emph{nearly-periodic}, and (b) they possess a Hamiltonian structure. Property (a) ensures the existence of the roto-rate vector, which may be thought of as an approximate $U(1)$-symmetry of the equations of motion. Property (b) enables the application of Noether's theorem to find an approximate conservation law, i.e. an adiabatic invariant, associated with this approximate symmetry. In order to explain and expand upon these points we will first discuss nearly-periodic systems that are not necessarily Hamiltonian. In particular we will derive a coordinate-free formula for the roto-rate vector associated with such a system. This discussion will form the content of Section \ref{section_roto_rate}. Then we will specialize to nearly-periodic systems that happen to possess (presymplectic) Hamiltonian structure. This specialization will ultimately lead to the formulas for the adiabatic invariant series in Section \ref{section_adiabatic}. As a way of illustrating the application of our formula we will use it in Section \ref{charged_particle_example} to compute the charged-particle adiabatic invariant, and again in Section \ref{field_line_invariant_sec} to derive a field-line adiabatic invariant for magnetic fields with field lines that are \emph{nearly closed}.

Readers who are interested in expressions for adiabatic invariants, but who are not interested in the derivation of such expressions may skip directly to Theorem \ref{main_result}. The relevant formulas are Eqs. \eqref{mu0_formula}-\eqref{mu3_formula}. Appendix \ref{app:how_to} provides the details of how to work with these formulas using index notation.

\section{Nearly-periodic systems and the roto-rate vector\label{section_roto_rate}}
A \emph{nearly-periodic system} is a two-timescale dynamical system whose short timescale dynamics is characterized by strictly periodic motion. Examples include masses conjoined by a stiff spring hung on the free end of a pendulum, and a charged particle in a strong magnetic field. For the sake of clarity the following definition of nearly-periodic systems will be useful.

\begin{definition}[nearly-periodic system]
A nearly-periodic system is a (possibly-infinite-dimensional) ordinary differential equation of the form $\dot{z} = \epsilon^{-1}V_\epsilon(z)$ with the following properties.
\begin{itemize}
\item The vector field $V_\epsilon$ depends smoothly on $\epsilon$ in a neighborhood of $0\in\mathbb{R}$.
\item The limiting vector field $V_0 = \omega_0\,\xi_0$. Here $\xi_0$ is a vector field with integral curves that are strictly periodic with period $2\pi$, and the \emph{frequency function} $\omega_0$ is a smooth, positive function that is constant along $\xi_0$'s integral curves.
\end{itemize}
\end{definition}

\begin{remark}
While the frequency function is not allowed to pass through zero, the vector field $\xi_0$ may do so. 
Therefore the limiting short timescale dynamics described by $V_0$ may have fixed points. In contrast, \cite{Kruskal_1962} requires that $V_0$ is nowhere vanishing. We have chosen to relax Kruskal's stronger assumption because his theory really only requires a non-vanishing frequency function. Moreover zeros of $V_0$ do occur in practice, and indicate the presence of a so-called slow manifold. (C.f. \citep{MacKay_2004}.)
\end{remark}

Away from the zeros of $\xi_0$, nearly-periodic systems exhibit a timescale separation that increases as $\epsilon$ tends to $0$. This suggests that averaging over the fast periodic motion described by $V_0$ ought to be permissible for small $\epsilon$. In more geometric terms, it is reasonable to expect that the equations of motion $\dot{z} = \epsilon^{-1}V_\epsilon(z)$ defining a nearly-periodic system possess an approximate $U(1)$-symmetry whose infinitesimal generator is given by $\xi_0$ to leading order in $\epsilon$. 

If the equations of motion possessed a \emph{true} $U(1)$-symmetry then there would be a vector field $\xi_\epsilon$ on $z$-space, which we will call $Z$, with the following properties.
\begin{enumerate}
\item The integrals curves of $\xi_\epsilon$, i.e. the solutions of the ODE $\dot{z} = \xi_\epsilon(z)$, must each be periodic with period $2\pi$.
\item The flows of $\xi_\epsilon$ and $V_\epsilon$ must commute. Equivalently, $[\xi_\epsilon,V_\epsilon] = 0$, where $[\cdot,\cdot]$ denotes the vector field commutator. 
\end{enumerate}
Such a $\xi_\epsilon$ is referred to as the infinitesimal generator of a $U(1)$-symmetry.

Given a nearly-periodic system the existence of such a $\xi_\epsilon$ is typically too much to hope for. On the other hand it is always possible to find a formal power series,
\[
\xi_\epsilon = \xi_0 + \epsilon\, \xi_1 + \epsilon^2\, \xi_2 + \dots,
\] 
whose coefficients $\xi_k$ are vector fields on $Z$, and that satisfies the properties (a) and (b) to all-orders in $\epsilon$. Such a formal power series is known as a \emph{roto-rate vector}. Existence of a roto-rate vector is one way to precisely define the notion of approximate $U(1)$-symmetry.

\begin{definition}[roto-rate vector]\label{roto_def}
Given a nearly-periodic system $\dot{z} = \epsilon^{-1}\,V_\epsilon(z)$, a \emph{roto-rate vector} is a formal power series $\xi_\epsilon = \xi_0 + \epsilon\, \xi_1 + \epsilon^2 \,\xi_2 + \dots$ with vector field coefficients such that $ \xi_0 = V_0/\omega_0$ and
\begin{itemize}
\item $[\xi_\epsilon,V_\epsilon] = 0$ 
\item $\mathrm{ln}\left(\exp(-2\pi\,\xi_0)\circ\exp(2\pi\,\xi_\epsilon)\right) = 0$,
\end{itemize}
where the previous two equalities are understood in the sense of formal power series.
\end{definition}

\begin{remark}
The integral curves of a vector field $\xi_\epsilon$ will be $2\pi$-periodic if and only if the exponential $\exp(2\pi \xi_\epsilon)$ is equal to the identity map on $z$-space. If $\xi_0$ happens to already have this property then it must be the case that $\text{id}_Z = \exp(2\pi\,\xi_0)\circ \exp(-2\pi\xi_0)\circ\exp(2\pi\xi_\epsilon) = \exp(-2\pi\xi_0)\circ\exp(2\pi\xi_\epsilon). $ By the Baker-Campbell-Hausdorff formula there is a formal power series vector field $Z_\epsilon$ such that 
\[\exp(Z_\epsilon) =  \exp(-2\pi\xi_0)\circ\exp(2\pi\xi_\epsilon),\] i.e. $Z_\epsilon = \mathrm{ln}\left(\exp(-2\pi\xi_0)\circ\exp(2\pi\xi_\epsilon)\right)$. Because $\xi_0$ is $\epsilon$-close to $\xi_\epsilon$ $Z_\epsilon$ must be $\epsilon$-small. The only formal power series $Z_\epsilon = Z_0 + \epsilon Z_1 + \dots$ that is $\epsilon$-small and that formally exponentiates to the identity is $Z_\epsilon = 0$. This explains the second property in the definition.
\end{remark}

Roto-rate vectors are remarkable due to the following.

\begin{theorem}[Existence and uniqueness of the roto-rate vector]\label{roto_rate_existence_thm}
Given a nearly-periodic system $\dot{z} = \epsilon^{-1} V_\epsilon(z)$ with $V_0 = \omega_0\,\xi_0$ there is a unique roto-rate vector $\xi_\epsilon$.
\end{theorem}

\begin{proof}
This result follows from minor modifications of the arguments in \cite{Kruskal_1962}, which does not allow $\xi_0$ to have fixed points. Therefore we will only outline the main steps in the proof. 

The first step is show that there is a (non-unique) formally-defined near-identity diffeomorphism $T_\epsilon:Z\rightarrow Z$ such that $\overline{V}_\epsilon = (T_\epsilon)_*V_\epsilon$ takes the form $\overline{V}_\epsilon = \overline{\omega}_\epsilon\,\xi_0 + \epsilon \,\delta \overline{V}_\epsilon$, where $\mathcal{L}_{\xi_0}\overline{\omega}_\epsilon = 0$ and $[\xi_0,\delta\overline{V}_\epsilon] = 0$. Note that (formally) pulling back this expression for $\overline{V}_\epsilon$ along $T_\epsilon$ implies $V_\epsilon = \omega_\epsilon\,\xi_\epsilon + \epsilon\,\delta V_\epsilon$, where $\omega_\epsilon = T^*_\epsilon \overline{\omega}_\epsilon$, $\xi_\epsilon = T_\epsilon^*\xi_0$, and $\delta V_\epsilon = T_\epsilon^*\overline{V}_\epsilon$. This establishes the existence of at least one roto-rate vector because $\xi_\epsilon$ apparently has $2\pi$-periodic integral curves, satisfies $\xi_0 = \xi_0$, and 
\[
[\xi_\epsilon, V_\epsilon] = \mathcal{L}_{\xi_\epsilon}(\omega_\epsilon\,\xi_\epsilon) + \epsilon\,\mathcal{L}_{\xi_\epsilon}\delta V_\epsilon = 0.
\]
A procedure for finding the diffeomorphism $T_\epsilon$ is the most commonly quoted result from  \cite{Kruskal_1962}. The reason the procedure still works when $\xi_0$ has fixed points is that solvability of the differential equations defining $T_\epsilon$ only requires periodicity of the $\xi_0$-flow and $\omega_{0}$ to be nowhere vanishing.

The second step is to show that if $\xi_\epsilon^\prime$ is any other roto-rate vector field then $\xi_\epsilon^\prime = \xi_\epsilon$. While it is less well-known, this argument is also contained in \cite{Kruskal_1962}. It proceeds along the following lines. Let $\overline{\xi}_\epsilon^\prime = T_{\epsilon*}\xi_\epsilon^\prime$. Introduce the decomposition $\overline{\xi}_\epsilon^\prime = \langle \overline{\xi}_\epsilon^\prime\rangle + (\overline{\xi}_\epsilon^\prime)^{\text{osc}}$, where $\langle \overline{\xi}_\epsilon^\prime\rangle = (2\pi)^{-1}\int_0^{2\pi} \exp(\theta\,\xi_0)^*\overline{\xi}_\epsilon\,d\theta$. Because $[\overline{\xi}_\epsilon^\prime,\overline{V}_\epsilon]=0$ it must also be the case that $[(\overline{\xi}_\epsilon^\prime)^{\text{osc}},\overline{V}_\epsilon]=0$, which in turn is equivalent to the sequence of conditions
\begin{align}
&[(\overline{\xi}_0^\prime)^{\text{osc}},\omega_0\,\xi_0] = 0\label{collapse_one}\\
&[(\overline{\xi}_1^\prime)^{\text{osc}},\omega_0\,\xi_0] +[(\overline{\xi}_0^\prime)^{\text{osc}}, \overline{V}_1] = 0\label{collapse_two}\\
&\dots\nonumber
\end{align}
The first condition \eqref{collapse_one} is satisfied if and only if $(\overline{\xi}_0^\prime)^{\text{osc}} = 0$. Substituting this in the second condition \eqref{collapse_two} therefore implies $[(\overline{\xi}_1^\prime)^{\text{osc}},\omega_0\,\xi_0] = 0$, which requires $(\overline{\xi}_1^\prime)^{\text{osc}} = 0$. This pattern continues to all orders in $\epsilon$ and shows that $(\overline{\xi}_\epsilon^\prime)^{\text{osc}}=0$. Now the argument may be completed as follows. Because $\overline{\xi}_\epsilon^\prime = \langle\overline{\xi}_\epsilon^\prime\rangle$ is $S^1$-invariant the difference $\overline{\xi}_\epsilon^\prime-\xi_0$ must also be $S^1$-invariant. Moreover because $\overline{\xi}_\epsilon^\prime$ and $\xi_0$ agree when $\epsilon=0$ there must be an $S^1$-invariant $O(1)$ vector field $w_\epsilon$ such that $\overline{\xi}_\epsilon^\prime - \xi_0 = \epsilon w_\epsilon$. Therefore $\exp(2\pi\,\overline{\xi}_\epsilon^\prime) = \exp(2\pi\, \xi_0 +2\pi\,\epsilon\,w_\epsilon ) = \exp(2\pi\,\xi_0)\circ\exp(2\pi\,\epsilon\,w_\epsilon) = \exp(2\pi\,\epsilon\,w_\epsilon) = \text{id}_Z$ in order for the integral curves of $\overline{\xi}_\epsilon^\prime$ to each be $2\pi$-periodic. (Note that we have made use of the commutativity $[\xi_0,w_\epsilon] = 0$.) This identity may only be satisfied if $w_\epsilon =0$.
\end{proof}

The preceding Theorem establishes the useful fact that by expanding the pair of conditions from Definition \ref{roto_def} in power series it should be possible to find the coefficients of the expansion $\xi_\epsilon = \xi_0 + \epsilon \,\xi_1  +\dots$ order-by-order. We will now follow this line of reasoning to derive explicit formulas for $\xi_0,\xi_1,\xi_2$, and $\xi_3$ in terms of Fourier harmonics of $V_\epsilon$ relative to $\xi_0$.

As a preparatory step we will establish the following variant of the BCH formula that is well-suited to perturbation theory in $\epsilon$.

\begin{lemma}[perturbative BCH formula]\label{BCH}
Let $A$ and $B$ be vector fields on the manifold $Z$ and $\epsilon$ a small real parameter. The logarithm $Z_\epsilon = \mathrm{ln}(\exp(-A)\circ\exp(A+\epsilon B))$ exists as a formal power series in $\epsilon$, $Z_\epsilon = Z_0 + \epsilon Z_1 + \epsilon^2 Z_2 + \dots$. The formulas
\begin{align}
Z_0 &= 0\label{Z0_formula}\\
Z_1 & =\int_0^1 B_{\tau_1}\,d\tau_1\\
Z_2 & = \frac{1}{2}\int_0^1\int_0^{\tau_1}[B_{\tau_2},B_{\tau_1}]\,d\tau_2\,d\tau_1\label{Z2_formula}\\
Z_3 & = \frac{1}{6}\int_0^1\int_0^{\tau_1}\int_0^{\tau_2}\bigg( [B_{\tau_3},[B_{\tau_2},B_{\tau_1}]] + [[B_{\tau_3},B_{\tau_2}],B_{\tau_1}]\bigg)\,d\tau_3\,d\tau_2\,d\tau_1,\label{Z3_formula}
\end{align}
with $B_\tau = \exp(\tau A)^* B$, give the first few coefficients $Z_k$. More generally
\begin{align}
Z_\epsilon = \epsilon \int_0^1\psi\left(\exp(\lambda \mathcal{L}_{A+\epsilon B})\exp(-\lambda\mathcal{L}_A)\right) \exp(\lambda A)^*B\,d\lambda\label{all_orders_bch}
\end{align}
with $\psi(z) = z\mathrm{ln}z/(z-1)$, gives $Z_\epsilon$ to all orders in $\epsilon$.
\end{lemma}

\begin{proof}
The proof proceeds by first solving a seemingly more-difficult problem, namely finding an asymptotic series representation for $Z_{\epsilon,\lambda} = \text{ln}(\exp(-\lambda A)\circ\exp(\lambda [A+\epsilon B]))$. To that end, first consider the $\lambda$-derivative of $\exp(Z_{\epsilon,\lambda}) = \exp(-\lambda A)\circ\exp(\lambda [A+\epsilon B])$,
\begin{align}
\partial_\lambda\exp(Z_{\epsilon,\lambda}) &= - A\circ \exp(Z_{\epsilon,\lambda}) + T\exp(-\lambda A)\circ [A+\epsilon B]\circ\exp(\lambda [A+\epsilon B])\nonumber\\
& = (-A + \exp(\lambda A)^*[A+\epsilon B])\circ \exp(Z_{\epsilon,\lambda}).\nonumber
\end{align}
In other words
\begin{align}
\partial_\lambda\exp(Z_{\epsilon,\lambda}) \circ \exp(-Z_{\epsilon,\lambda}) = \epsilon \exp(\lambda A)^* B.\label{proto_all_orders_bch}
\end{align} 
We will eventually obtain \eqref{all_orders_bch} by integrating \eqref{proto_all_orders_bch} in $\lambda$, but first we need an expression for $\partial_\lambda\exp(Z_{\epsilon,\lambda}) \circ \exp(-Z_{\epsilon,\lambda})$  in terms of $\partial_\lambda Z_{\epsilon,\lambda}$. One way to find such an expression is the following. Let $C_\lambda$ be any $\lambda$-dependent vector field and set  $\psi_{s,\lambda} = \exp(sC_\lambda)$. By the equality of mixed partials the vector fields $V_{s,\lambda} = \partial_\lambda \psi_{s,\lambda}\circ \psi_{s,\lambda}^{-1} $ and $\xi_{s,\lambda} = \partial_s \psi_{s,\lambda}\circ \psi_{s,\lambda}^{-1}  = C_\lambda$ must be related by the condition
\begin{align}
\partial_s V_{s,\lambda} + \mathcal{L}_{C_\lambda} V_{s,\lambda} = \partial_\lambda C_\lambda.
\end{align}
Thinking of the last condition as a differential equation for $V_{s,\lambda}$, it can be solved using the method of variation of parameters. The solution for $V_{s,\lambda}$ is given by
\begin{align}
V_{s,\lambda}& = \exp(-s C_\lambda)^*\int_0^s \exp(\overline{s} C_\lambda)^*\partial_\lambda C_\lambda\,d\overline{s}.\label{gen_exp_id}
\end{align}
Because $V_{1,\lambda} = \partial_\lambda \psi_{1,\lambda}\circ \psi_{1,\lambda}^{-1} = \partial_\lambda\exp(C_\lambda)\circ \exp(-C_\lambda)$ Eq.\,\eqref{gen_exp_id} implies the general formula
\begin{align}
\partial_\lambda\exp(C_\lambda)\circ \exp(-C_\lambda) &= \exp(- C_\lambda)^*\int_0^1 \exp(\overline{s} C_\lambda)^*\partial_\lambda C_\lambda\,d\overline{s}\nonumber\\
& = \phi(-\mathcal{L}_{C_\lambda}) \partial_\lambda C_\lambda,\label{useful_derivaitive_formula_exp}
\end{align}
where $\phi(z) = [\exp(z) - 1]/z$. Applying this formula to \eqref{proto_all_orders_bch} then gives
\begin{align}
&\phi(-\mathcal{L}_{Z_{\epsilon,\lambda}}) \partial_\lambda Z_{\epsilon,\lambda} =  \epsilon \exp(\lambda A)^* B\nonumber\\
\Rightarrow & \partial_\lambda Z_{\epsilon,\lambda}=\epsilon\,\frac{1}{\phi(-\mathcal{L}_{Z_{\epsilon,\lambda}})} \exp(\lambda A)^* B\nonumber\\
\Rightarrow & Z_{\epsilon,\lambda} = \int_0^{\lambda}\epsilon\,\frac{1}{\phi(-\mathcal{L}_{Z_{\epsilon,\overline{\lambda}}})} \exp(\overline{\lambda} A)^* B\,d\overline{\lambda}.\label{apparently_not_useful}
\end{align}
While \eqref{apparently_not_useful} may not seem helpful because $Z_{\epsilon,\lambda}$ appears under the integral sign, in fact it implies \eqref{all_orders_bch} for the following reason. Because 
\[
\exp(\mathcal{L}_{Z_{\epsilon,\lambda}}) = (\exp(-\lambda A)\circ \exp(\lambda[A+\epsilon B]))^* = \exp(\lambda \mathcal{L}_{A+\epsilon B})\exp(-\lambda \mathcal{L}_{A}),
\]
the Lie derivative $\mathcal{L}_{Z_{\epsilon,\lambda}} $ may be written
\[
\mathcal{L}_{Z_{\epsilon,\lambda}} = \text{ln}\left(\exp(\lambda \mathcal{L}_{A+\epsilon B})\exp(-\lambda \mathcal{L}_{A})\right).
\]
If $a \equiv \exp(\lambda \mathcal{L}_{A+\epsilon B})\exp(-\lambda \mathcal{L}_{A})$ it therefore follows that
\begin{align}
\frac{1}{\phi(-\mathcal{L}_{Z_{\epsilon,\lambda}})} &= \frac{1}{\phi(-\text{ln}a)}\nonumber\\
& = -\frac{\text{ln}a}{\exp(-\text{ln}a) - 1}\nonumber\\
& = \psi(a).\label{single_variable_trick}
\end{align}
Substituting \eqref{single_variable_trick} in \eqref{apparently_not_useful} gives \eqref{all_orders_bch}, as desired.

In order to obtain the formulas \eqref{Z0_formula}-\eqref{Z3_formula} it is sufficient to expand the formal expression \eqref{all_orders_bch} as a power series in $\epsilon$. This rather tedious calculation proceeds as follows. First it is useful to find the power series expansion of the operator $a_{\epsilon,\lambda} = \exp(\lambda \mathcal{L}_{A+\epsilon B})\exp(-\lambda \mathcal{L}_{A})$. Let $f:Z\rightarrow\mathbb{R}$ be any scalar on $Z$ and introduce $f_\lambda =  \exp(\lambda \mathcal{L}_{A+\epsilon B}) f$. The scalar $f_\lambda$ obeys the differential equation $\partial_\lambda f_\lambda = \mathcal{L}_{A+\epsilon B} f_\lambda$. Introducing the variation-of-parameters ansatz $f_\lambda = \exp(\lambda \mathcal{L}_{A})\overline{f}_\lambda$, the scalar $\overline{f}_\lambda$ therefore satisfies $\partial_\lambda \overline{f}_\lambda = \epsilon \mathcal{L}_{\exp(-\lambda A)^* B}\overline{f}_\lambda$, or in integral form 
\begin{align}
\overline{f}_\lambda &= f + \epsilon\int_0^\lambda \mathcal{L}_{B_{-s_1}}\overline{f}_{s_1}\,ds_1\nonumber\\
& = f+ \epsilon \int_0^\lambda \mathcal{L}_{B_{-s_1}}f\,ds_1 + \epsilon^2\int_0^\lambda\int_0^{s_1}\mathcal{L}_{B_{-s_1}}\mathcal{L}_{B_{-s_2}}f\,ds_2\,ds_1 + O(\epsilon^3),
\end{align}
where we have introduced the shorthand notation $B_{s} = \exp(s A)^* B$.
This shows that $a_{\epsilon,\lambda}$ has the asymptotic expansion 
\[
a_{\epsilon,\lambda} = \exp(\lambda \mathcal{L}_A)\left(1 + \epsilon \overline{a}_{1,\lambda} + \epsilon^2 \overline{a}_{2,\lambda} + \dots\right)\exp(-\lambda \mathcal{L}_A),\]
where
\begin{align}
\overline{a}_{1,\lambda} &=  \int_0^\lambda \mathcal{L}_{B_{-s_1}}\,ds_1\\
\overline{a}_{2,\lambda} & = \int_0^\lambda\int_0^{s_1}\mathcal{L}_{B_{-s_1}}\mathcal{L}_{B_{-s_2}}\,ds_2\,ds_1.
\end{align}
Combining this observation with the series representation of $\psi(1 + x) = 1 + \frac{1}{2}x - \frac{1}{6} x^2 + \frac{1}{12}x^3 + \dots $ therefore implies
\begin{align}
Z_0 &= 0\\
Z_1 & = \int_0^1 \exp(\lambda A)^*B\,d\lambda\nonumber\\
&= \int_0^1 \exp(\tau_1 A)^*B\,d\tau_1\\
Z_2 & = \frac{1}{2}\int_0^1 \exp(\lambda A)^* \overline{a}_{1,\lambda}B\,d\lambda\nonumber\\
& = \frac{1}{2}\int_0^1\int_0^\lambda \exp(\lambda A)^*[B_{-s_1},B]\,ds_1\,d\lambda\nonumber\\
& = \frac{1}{2}\int_0^1\int_0^{\tau_1}[B_{\tau_2},B_{\tau_1}]\,d\tau_2\,d\tau_1\\
Z_3 & = -\frac{1}{6}\int_0^1  \exp(\lambda A)^* \overline{a}_{1,\lambda}^2 B\,d\lambda\nonumber\\
& \quad+ \frac{1}{2}\int_0^1 \exp(\lambda A)^* \overline{a}_{2,\lambda}B\,d\lambda.\label{proto_Z3}
\end{align}
These expressions for $Z_0,Z_1,Z_2$ clearly reproduce \eqref{Z0_formula}-\eqref{Z2_formula}. To see that \eqref{proto_Z3} reproduces \eqref{Z3_formula} notice first that
\begin{align}
Z_3 & = -\frac{1}{6}\int_0^1  \exp(\lambda A)^*\int_0^\lambda\int_0^\lambda[ B_{-s_1},[B_{-s_2},B]]\,ds_2\,ds_1\,d\lambda\nonumber\\
& \quad+ \frac{1}{2}\int_0^1 \exp(\lambda A)^*  \int_0^\lambda\int_0^{s_1}[B_{-s_1}[B_{-s_2},B]]\,ds_2\,ds_1\,d\lambda.
\end{align}
Next observe that if $g(s_1,s_2) = [B_{-s_1},[B_{-s_2},B]]$ then by Fubini's theorem
\begin{align}
\int_0^\lambda\int_0^\lambda g(s_1,s_2)\,ds_2\,ds_1 = \int_0^\lambda\int_0^{s_1} g(s_1,s_2)\,ds_2\,ds_1 +\int_0^\lambda\int_0^{s_1} g(s_2,s_1)\,ds_2\,ds_1.
\end{align}
It follows that 
\begin{align*}
Z_3 & = -\frac{1}{6}\int_0^1  \exp(\lambda A)^*\int_0^\lambda\int_0^{s_1}([ B_{-s_1},[B_{-s_2},B]] +[ B_{-s_2},[B_{-s_1},B]]  )\,ds_2\,ds_1\,d\lambda\nonumber\\
& \quad+ \frac{1}{2}\int_0^1 \exp(\lambda A)^*  \int_0^\lambda\int_0^{s_1}[B_{-s_1}[B_{-s_2},B]]\,ds_2\,ds_1\,d\lambda\\
 &= \int_0^1\int_0^\lambda\int_0^{s_1}\exp(\lambda A)^*  \left(\frac{1}{3}[ B_{-s_1},[B_{-s_2},B]] - \frac{1}{6}[ B_{-s_2},[B_{-s_1},B]]\right)\,ds_2\,ds_1\,d\lambda\\
 & = \int_0^1\int_0^\lambda\int_0^{s_1}\exp(\lambda A)^*  \left(\frac{1}{6}[ B_{-s_1},[B_{-s_2},B]] + \frac{1}{6}[[B_{-s_1}, B_{-s_2}],B]]\right)\,ds_2\,ds_1\,d\lambda\nonumber\\
 & = \frac{1}{6}\int_0^1\int_0^{\tau_1}\int_0^{\tau_2}\bigg( [B_{\tau_3},[B_{\tau_2},B_{\tau_1}]] + [[B_{\tau_3},B_{\tau_2}],B_{\tau_1}]\bigg)\,d\tau_3\,d\tau_2\,d\tau_1
\end{align*}
where we have applied the Jacobi identity $[B_{-s_2},[B_{-s_1},B]] = [[B_{-s_2},B_{-s_1}],B] + [B_{-s_1},[B_2,B]]$ on the second-to-last line, and we changed integration variables to $\tau_1 = \lambda$, $\tau_3 = \lambda - s_1$, and $\tau_2 = \lambda - s_2$ on the last line.
\end{proof}

With the modified BCH formula from Lemma \ref{BCH} in hand it is now straightforward to derive formulas for the coefficients of $\xi_\epsilon = \xi_0 +\epsilon\,\xi_1 + \epsilon^2\,\xi_2 + \dots$ as follows. 

\begin{definition}[Mean and oscillating subspaces]
Given a nearly-periodic system with roto-rate vector $\xi_\epsilon$, the space of \emph{limiting mean vector fields} $\langle\mathfrak{X}(Z)\rangle$ or just \emph{mean vector fields} for short is the subspace of vector fields $A$ on $Z$ that are equal to their $U(1)$-average along $\xi_0$. In symbols $A\in \langle\mathfrak{X}(Z)\rangle$ means $A =\langle A\rangle (2\pi)^{-1}\int_0^{2\pi} \exp(\theta \,\xi_0)^*A\,d\theta$. The space of \emph{limiting oscillating vector fields} $\mathfrak{X}(Z)^{\text{osc}}$, or just \emph{oscillating vector fields} for short, is the subspace of vector fields on $Z$ that average to zero along $\xi_0.$ That is, $A\in \mathfrak{X}(Z)^{\text{osc}}$ if $\langle A\rangle = 0$. 
\end{definition}

\begin{remark}
Standard results on Fourier series imply that the mean and fluctuating subspaces are complimentary subspaces of $\mathfrak{X}(Z)$, the space of vector fields on $Z$. A projection onto $\langle\mathfrak{X}(Z)\rangle$ is $\overline{\pi}: A\mapsto \langle A\rangle$ and a projection onto $\mathfrak{X}(Z)^{\text{osc}}$ is $\widetilde{\pi} = 1 - \overline{\pi}$. If $A$ is any vector field on $Z$ then the notations $A = \langle A\rangle + A^{\mathrm{osc}}$ and $A = \langle A\rangle+\widetilde{A}$ will be used interchangeably to denote the decomposition of $A$ into its mean, $\langle A\rangle = \overline{\pi}A$, and fluctuating parts, $A^{\mathrm{osc}} = \widetilde{A} = \widetilde{\pi}A$.
\end{remark}

\begin{theorem}[formula for the roto-rate vector]\label{roto_formula_thm}
The first four coefficients of the roto-rate vector $\xi_\epsilon$ associated with a nearly-periodic system $\dot{z} = \epsilon^{-1}V_\epsilon(z)$ are given in terms of the power series expansion of $V_\epsilon = \omega_0 \xi_0 + \epsilon V_1 + \epsilon^2 V_2 + \dots$ as follows. 
\begin{align}
\xi_0 & = V_0/\omega_0\\
\xi_1 &= \Lxio I_0\widetilde{V}_1\label{xi1_formula}\\
\xi_2 & = \Lxio I_0\widetilde{V}_2 + \Lxio I_0[ I_0\widetilde{V}_1, \langle V_1\rangle] + \frac{1}{2}\Lxio I_0[I_0\widetilde{V}_1,\widetilde{V}_1]^{\text{osc}} + \frac{1}{2}[\Lxio I_0\widetilde{V}_1, I_0\widetilde{V}_1]\label{xi2_formula}\\
\xi_3 &= \Lxio\bigg(I_0\widetilde{V}_3 +  I_0[ I_0\widetilde{V}_1,\langle V_2\rangle]^{\text{osc}} +I_0[ I_0\widetilde{V}_2,\langle V_1\rangle]^{\text{osc}}\nonumber\\
&\quad \hphantom{\Lxio}+\frac{1}{2} I_0[I_0\widetilde{V}_2,\widetilde{V}_1]^{\text{osc}}+ \frac{1}{2}I_0[I_0\widetilde{V}_1,\widetilde{V}_2]^{\text{osc}}+\frac{1}{3}I_0[I_0\widetilde{V}_1,[I_0\widetilde{V}_1,\widetilde{V}_1]]^{\text{osc}}\nonumber\\
&\quad\hphantom{\Lxio}+I_0[  I_0[ I_0\widetilde{V}_1, \langle V_1\rangle] ,\langle V_1\rangle]
+ \frac{1}{2}I_0[    I_0[I_0\widetilde{V}_1,\widetilde{V}_1]^{\text{osc}} ,\langle V_1\rangle]\nonumber\\
 &\quad\hphantom{\Lxio}+\frac{1}{2}I_0[I_0\widetilde{V}_1,[I_0\widetilde{V}_1,\langle V_1\rangle]]^{\text{osc}}\bigg)\nonumber\\
 & +  \frac{1}{2} [\Lxio I_0 \widetilde{V}_1,I_0\widetilde{V}_2]+  \frac{1}{2} [\Lxio I_0 \widetilde{V}_2,I_0\widetilde{V}_1]\nonumber\\
 &+[I_0[\Lxio I_0\widetilde{V}_1,{V}_1]^{\text{osc}},I_0\widetilde{V}_1]+I_0[\langle [\Lxio I_0\widetilde{V}_1,\widetilde{V}_1]\rangle, I_0\widetilde{V}_1]\nonumber\\
 &+\frac{1}{3}[I_0\widetilde{V}_1,[\Lxio I_0\widetilde{V}_1,I_0\widetilde{V}_1]]
.\label{xi3_formula}
\end{align}
Here $I_0$ is the inverse of $\mathcal{L}_{V_0}$ restricted to the fluctuating subspace regarded as a linear map $\mathfrak{X}^{\mathrm{osc}}(Z)\rightarrow \mathfrak{X}^{\mathrm{osc}}(Z)$.
\end{theorem}

\begin{proof}
The proof proceeds by directly analyzing the conditions in Definition \ref{roto_def} order-by-order in $\epsilon$. First Lemma \ref{BCH} will be applied with $A = 2\pi \xi_0$ and $B = 2\pi (\xi_1 + \epsilon \xi_2 + \epsilon^2 \xi_3 + \dots)$ in order to identify the coefficients of the formal power series \[Z_{\epsilon} = \mathrm{ln}\left(\exp(-2\pi\,\xi_0)\circ\exp(2\pi\,\xi_\epsilon)\right).\] Then the power series coefficients of $[\xi_\epsilon,V_\epsilon]$ and $Z_\epsilon$ will each be set equal to zero.

After changing integration variables from $\tau_k$ to $\theta_k = 2\pi \tau_k$ and accounting for the fact that $B = 2\pi (\xi_1 + \epsilon \xi_2 + \epsilon^2 \xi_3 + \dots)$ is itself a formal power series, the first several coefficients of $Z_\epsilon$ given by Lemma \ref{BCH} are
\begin{align}
Z_0 &= 0\label{Z0_formula_roto_thm}\\
Z_1 & =2\pi \langle \xi_1\rangle \label{Z1_formula_roto_thm}\\
Z_2 & = 2\pi \langle \xi_2\rangle+\frac{1}{2}\int_0^{2\pi}\int_0^{\theta_1}[\xi_1^{\theta_2},\xi_1^{\theta_1}]\,d\theta_2\,d\theta_1\label{Z2_formula_roto_thm}\\
Z_3 & =  2\pi \langle \xi_3\rangle+\frac{1}{2}\int_0^{2\pi}\int_0^{\theta_1}[\xi_1^{\theta_2},\xi_2^{\theta_1}]\,d\theta_2\,d\theta_1+ \frac{1}{2}\int_0^{2\pi}\int_0^{\theta_1}[\xi_2^{\theta_2},\xi_1^{\theta_1}]\,d\theta_2\,d\theta_1\nonumber\\
&+\frac{1}{6}\int_0^{2\pi}\int_0^{\theta_1}\int_0^{\theta_2}\bigg( [\xi_1^{\theta_3},[\xi_1^{\theta_2},\xi_1^{\theta_1}]] + [[\xi_1^{\theta_3},\xi_1^{\theta_2}],\xi_1^{\theta_1}]\bigg)\,d\theta_3\,d\theta_2\,d\theta_1.\label{Z3_formula_roto_thm}
\end{align}
where $\xi_k^{\theta_j} = \exp(\theta_j\,\xi_0)^*\xi_k$. Each of these coefficients must vanish, but we will not examine the consequences of this vanishing now. Instead we will examine the vanishing of the $Z_k$ and the coefficients of $[\xi_\epsilon,V_\epsilon]$ incrementally and simultaneously in the following paragraphs.

The $O(1)$ coefficients of the series $Z_\epsilon$ and $[\xi_\epsilon,V_\epsilon]$ are given by Eq.\,\eqref{Z0_formula_roto_thm} and $[\xi_0,V_0]$, respectively. The former is obviously zero, while the latter vanishes because $\mathcal{L}_{\xi_0}\omega_0 = 0$. Thus no constraints are placed on the $\xi_k$ at this order. Note that $\xi_0 = V_0/\omega_0$ \emph{by definition} of the roto-rate vector.

The $O(\epsilon)$ coefficients of $Z_\epsilon$ and $[\xi_\epsilon,V_\epsilon]$ are given by Eq.\,\eqref{Z1_formula_roto_thm} and $[\xi_0,V_1] + [\xi_1,V_0]$, respectively. Vanishing of these coefficients is equivalent to the joint satisfaction of the three conditions
\begin{align}
0& = \langle \xi_1 \rangle\label{Z1_vanishing_roto_thm}\\
0& = [\xi_0,V_1] + [\xi_1,V_0]^{\text{osc}}\label{commutator1_osc_vanishing_roto_thm}\\
0& = \langle [\xi_1,V_0] \rangle\label{commutator1_mean_vanishing_roto_thm}.
\end{align}
We claim that the conditions \eqref{Z1_vanishing_roto_thm} and \eqref{commutator1_osc_vanishing_roto_thm} uniquely determine $\xi_1$, and that when $\xi_1$ is so determined the condition \eqref{commutator1_mean_vanishing_roto_thm} is satisfied automatically. As for the first part of our claim, notice that condition \eqref{commutator1_osc_vanishing_roto_thm} is equivalent to the linear equation $\mathcal{L}_{V_0}\widetilde{\xi}_1 = \Lxio \widetilde{V}_1$, which has the unique solution $\widetilde{\xi}_1 = \Lxio I_0\widetilde{V}_1$. Because condition \eqref{Z1_vanishing_roto_thm} says that $\xi_1$ has zero average, the last observation implies that in fact $\xi_1 = \Lxio I_0\widetilde{V}_1$, which is precisely the desired formula \eqref{xi1_formula}. As for the second part of our claim, it is enough to observe that, because $\xi_1 = \widetilde{\xi}_1$, $\langle [\xi_1,V_0] \rangle = \langle [\widetilde{\xi_1},V_0] \rangle = [\langle \widetilde{\xi}_1\rangle , V_0] = 0$.

The $O(\epsilon^2)$ coefficients of $Z_\epsilon$ and $[\xi_\epsilon, V_\epsilon]$ are given by Eq.\,\eqref{Z1_formula_roto_thm} and $[\xi_0,V_2] + [\xi_1,V_1] + [\xi_2,V_0]$, respectively. Vanishing of these coefficients is equivalent to
\begin{align}
0& =  \langle \xi_2\rangle+\frac{1}{2}\fint\int_0^{\theta_1}[\xi_1^{\theta_2},\xi_1^{\theta_1}]\,d\theta_2\,d\theta_1 \label{Z2_vanishing_roto_thm}\\
0& = [\xi_0,V_2] + [\xi_1,V_1]^{\text{osc}} + [\xi_2,V_0]^{\text{osc}}\label{commutator2_osc_vanishing_roto_thm}\\
0& = \langle [\xi_1,V_1]\rangle + \langle [\xi_2,V_0] \rangle.\label{commutator2_mean_vanishing_roto_thm}
\end{align}
As was the case with the $O(\epsilon)$ coefficients, we claim that conditions \eqref{Z2_vanishing_roto_thm} and \eqref{commutator2_osc_vanishing_roto_thm} uniquely determine $\xi_2$, and that, with $\xi_2$ so determined, condition \eqref{commutator2_mean_vanishing_roto_thm} is satisfied automatically. First observe that condition \eqref{Z2_vanishing_roto_thm} completely determines $\langle \xi_2\rangle$. Indeed, using Eq.\,\eqref{xi1_formula} inside of the integral and recognizing $\Lxio I_0 \widetilde{V}_1^{\theta_2} = \partial_{\theta_2} I_0 \widetilde{V}_1^{\theta_2}$ leads to
\begin{align}
\langle \xi_2\rangle& =  -\frac{1}{2}\fint\int_0^{\theta_1}[\Lxio I_0 \widetilde{V}_1^{\theta_2},\Lxio I_0 \widetilde{V}_1^{\theta_1}]\,d\theta_2\,d\theta_1\nonumber\\
& = -\frac{1}{2}\fint\int_0^{\theta_1}[\partial_{\theta_2} I_0 \widetilde{V}_1^{\theta_2},\Lxio I_0 \widetilde{V}_1^{\theta_1}]\,d\theta_2\,d\theta_1\nonumber\\
& =  -\frac{1}{2}\fint[ I_0 \widetilde{V}_1^{\theta_1},\Lxio I_0 \widetilde{V}_1^{\theta_1}]\,d\theta_1+ \frac{1}{2}\fint[ I_0 \widetilde{V}_1^{},\Lxio I_0 \widetilde{V}_1^{\theta_1}]\,d\theta_2\,d\theta_1\nonumber\\
& = \frac{1}{2} \langle [ \Lxio I_0 \widetilde{V}_1^{},I_0 \widetilde{V}_1^{}] \rangle.\label{xi2_mean_proof}
\end{align}
Next observe that condition \eqref{commutator2_osc_vanishing_roto_thm} is equivalent to the linear equation
\begin{align}
\mathcal{L}_{V_0}\widetilde{\xi}_2 &=  [\xi_0,V_2] + [\xi_1,V_1]^{\text{osc}}\nonumber\\
& = \Lxio \widetilde{V}_2  + \Lxio [ I_0\widetilde{V}_1, \langle V_1\rangle] + [\Lxio I_0 \widetilde{V}_1,\widetilde{V}_1]^{\text{osc}},\label{a_roto_thm}
\end{align}
which has the unique solution
\begin{align}
\widetilde{\xi}_2 & = \Lxio\bigg(I_0\widetilde{V}_2 + I_0[ I_0\widetilde{V}_1, \langle V_1\rangle]\bigg) + I_0 [\Lxio I_0 \widetilde{V}_1,\widetilde{V}_1]^{\text{osc}}\nonumber\\
& = \Lxio\bigg(I_0\widetilde{V}_2 + I_0[ I_0\widetilde{V}_1, \langle V_1\rangle] + \frac{1}{2}I_0[I_0\widetilde{V}_1,\widetilde{V}_1]^{\text{osc}}\bigg) + \frac{1}{2}[\Lxio I_0\widetilde{V}_1, I_0\widetilde{V}_1]^{\text{osc}}.\label{b_roto_thm}
\end{align}
On the second line of Eq.\,\eqref{b_roto_thm} we have used the identity
\begin{align}
I_0[\Lxio I_0\widetilde{V}_1,\widetilde{V}_1]^{\text{osc}} = \frac{1}{2}[\Lxio I_0\widetilde{V}_1,I_0\widetilde{V}_1]^{\text{osc}}  + \frac{1}{2}\Lxio I_0[ I_0\widetilde{V}_1, \widetilde{V}_1]^{\text{osc}},\label{two_term_id_roto_thm}
\end{align}
which follows from the nontrivial recursive relationship
\begin{align}
I_0[\Lxio I_0\widetilde{V}_1,\widetilde{V}_1]^{\text{osc}} &=I_0[\Lxio I_0\widetilde{V}_1,\mathcal{L}_{V_0}I_0\widetilde{V}_1]^{\text{osc}}\nonumber\\
& = [\Lxio I_0\widetilde{V}_1,I_0\widetilde{V}_1]^{\text{osc}} - I_0[\Lxio \widetilde{V}_1,I_0\widetilde{V}_1]^{\text{osc}}\nonumber\\
& = [\Lxio I_0\widetilde{V}_1,I_0\widetilde{V}_1]^{\text{osc}} + \Lxio I_0[ I_0\widetilde{V}_1, \widetilde{V}_1]^{\text{osc}} - I_0[\Lxio I_0\widetilde{V}_1, \widetilde{V}_1]^{\text{osc}}.
\end{align}
Adding Eqs.\,\eqref{xi2_mean_proof} and \eqref{b_roto_thm} demonstrates the first part of our claim, in addition to giving the desired formula \eqref{xi2_formula} for $\xi_2$. As for the second part of our claim, the expression \eqref{xi2_mean_proof} for $\langle\xi_2\rangle$ implies
\begin{align}
\langle [\xi_1,V_1]\rangle + \langle [\xi_2,V_0] \rangle & = \langle [ \Lxio I_0\widetilde{V}_1,\widetilde{V}_1]\rangle+ \langle [\langle\xi_2\rangle,V_0] \rangle\nonumber\\
& = \langle [ \Lxio I_0\widetilde{V}_1,\widetilde{V}_1]\rangle + \frac{1}{2}\langle [[\Lxio I_0\widetilde{V}_1, I_0\widetilde{V}_1],V_0] \rangle\nonumber\\
&  = \langle [ \Lxio I_0\widetilde{V}_1,\widetilde{V}_1]\rangle -\frac{1}{2}\langle [\Lxio \widetilde{V}_1, I_0\widetilde{V}_1] \rangle - \frac{1}{2}\langle [\Lxio I_0\widetilde{V}_1, \widetilde{V}_1] \rangle \nonumber\\
& = 0,
\end{align}
as claimed.

The pattern established at the previous orders in $\epsilon$ continues with the $O(\epsilon^3)$ coefficients of $Z_\epsilon$ and $[\xi_\epsilon, V_\epsilon]$. Vanishing of the third-order coefficients is equivalent to the trio of conditions
\begin{align}
\langle \xi_3\rangle  & = -\frac{1}{2}\fint \int_0^{\theta_1} [\xi_1^{\theta_2},\xi_2^{\theta_1}]\,d\theta_2\,d\theta_1-\frac{1}{2}\fint \int_0^{\theta_1} [\xi_2^{\theta_2},\xi_1^{\theta_1}]\,d\theta_2\,d\theta_1\nonumber\\
& - \frac{1}{6}\fint \int_0^{\theta_1}\int_0^{\theta_2} \bigg([\xi_1^{\theta_3},[\xi_1^{\theta_2},\xi_1^{\theta_1}]]+[[\xi_1^{\theta_3},\xi_1^{\theta_2}],\xi_1^{\theta_1}]\bigg)\,d\theta_3\,d\theta_2\,d\theta_1\label{Z3_vanishing_roto_thm}\\
0 & = [\xi_0,\widetilde{V}_3] + [\xi_1,V_2]^{\text{osc}} + [\xi_2,V_1]^{\text{osc}} + [\xi_3,V_0]\label{commutator3_osc_vanishing_roto_thm}\\
0& =  \langle [\xi_1,V_2]\rangle + \langle [\xi_2,V_1]\rangle + [\langle \xi_3\rangle ,V_0].\label{commutator3_mean_vanishing_roto_thm}
\end{align}
To see that Eq.\,\eqref{Z3_vanishing_roto_thm} determines $\langle \xi_3\rangle$ first use Fubini's theorem and the Lie derivative formula to simplify the double integrals as
\begin{align}
\langle \xi_3\rangle  & = \langle [\widetilde{\xi}_2,I_0\widetilde{V}_1]\rangle + [I_0\widetilde{V}_1,\langle \xi_2\rangle]\nonumber\\
& - \frac{1}{6}\fint \int_0^{\theta_1}\int_0^{\theta_2} \bigg([\xi_1^{\theta_3},[\xi_1^{\theta_2},\xi_1^{\theta_1}]]+[[\xi_1^{\theta_3},\xi_1^{\theta_2}],\xi_1^{\theta_1}]\bigg)\,d\theta_3\,d\theta_2\,d\theta_1.
\end{align}
Next use the same techniques to perform the $\theta_3$ and $\theta_1$ integrations in the triple integral according to
\begin{align}
\langle \xi_3\rangle  & = \langle [\widetilde{\xi}_2,I_0\widetilde{V}_1]\rangle + [I_0\widetilde{V}_1,\langle \xi_2\rangle]\nonumber\\
& = \langle [\widetilde{\xi}_2,I_0\widetilde{V}_1]\rangle + [I_0\widetilde{V}_1,\langle \xi_2\rangle]- \frac{1}{3}\langle[[\Lxio I_0\widetilde{V}_1,I_0\widetilde{V}_1],I_0\widetilde{V}_1] \rangle   \nonumber\\
& - \frac{1}{3}\fint \int_0^{\theta_1} \bigg([I_0\widetilde{V}_1^{\theta_2},[\xi_1^{\theta_2},I_0\widetilde{V}_1^{}]]+[[I_0\widetilde{V}_1^{\theta_2},\xi_1^{\theta_2}],I_0\widetilde{V}_1^{}]\bigg)\,d\theta_3\,d\theta_2.
\end{align}
Finally apply the identity 
\begin{align}
[I_0\widetilde{V}_1^{\theta_2},[\xi_1^{\theta_2},I_0\widetilde{V}_1^{}]] = \frac{1}{2}\partial_{\theta_2} [I_0\widetilde{V}_1^{\theta_2},[I_0\widetilde{V}_1^{\theta_2},I_0\widetilde{V}_1]] -\frac{1}{2}[[\Lxio I_0\widetilde{V}_1^{\theta_2},I_0\widetilde{V}_1^{\theta_2}],I_0\widetilde{V}_1]
\end{align}
to obtain
\begin{align}
\langle \xi_3\rangle & = \langle [\widetilde{\xi}_2,I_0\widetilde{V}_1]\rangle + [I_0\widetilde{V}_1,\langle \xi_2\rangle]- \frac{1}{3}\langle[[\Lxio I_0\widetilde{V}_1,I_0\widetilde{V}_1],I_0\widetilde{V}_1] \rangle    \nonumber\\
& + \frac{1}{2} [\langle[\Lxio I_0\widetilde{V}_1,I_0\widetilde{V}_1]\rangle,I_0\widetilde{V}_1]\nonumber\\
& = \langle [\widetilde{\xi}_2,I_0\widetilde{V}_1]\rangle - \frac{1}{3}\langle[[\Lxio I_0\widetilde{V}_1,I_0\widetilde{V}_1],I_0\widetilde{V}_1] \rangle\nonumber\\
& = \langle [\Lxio I_0\widetilde{V}_2,I_0\widetilde{V}_1]\rangle + \langle [\Lxio I_0[I_0\widetilde{V}_1,\langle V_1\rangle],I_0\widetilde{V}_1] \rangle  \nonumber\\
& + \frac{1}{2}\langle [\Lxio I_0[I_0\widetilde{V}_1,\widetilde{V}_1]^{\text{osc}},I_0\widetilde{V}_1] \rangle + \frac{1}{6} \langle[[\Lxio I_0\widetilde{V}_1,I_0\widetilde{V}_1],I_0\widetilde{V}_1]\rangle.\label{xi3_mean}
\end{align}
For the oscillating part of $\xi_3$ use Eq.\,\eqref{commutator3_osc_vanishing_roto_thm} to obtain the general formula
\begin{align}
\widetilde{\xi}_3 &= I_0[\xi_0,\widetilde{V}_3] + I_0[\xi_1,V_2]^{\text{osc}} + I_0[\xi_2,V_1]^{\text{osc}}.\label{xi3_general_formula}
\end{align}
Using Eq.\,\eqref{xi1_formula} for $\xi_1$ and Eq.\,\eqref{xi2_formula} for $\xi_2$ this formula for $\widetilde{\xi}_3$ may be added to Eq.\,\eqref{xi3_mean} and then manipulated so as to yield \eqref{xi3_formula}. The details of this tedious calculation may be found in Appendix \ref{appA}. The proof will now be complete as soon as we show that if $\xi_1,\xi_2$, and $\xi_3$ are given by Eqs.\,\eqref{xi1_formula}-\eqref{xi3_formula}, respectively, then condition \eqref{commutator3_mean_vanishing_roto_thm} is satisfied automatically. This may be seen by the following direct calculation with $I = \langle [\xi_1,V_2]\rangle + \langle [\xi_2,V_1]\rangle + [\langle \xi_3\rangle ,V_0]$,
\begin{align}
I & = \langle[\mathcal{L}_{\xi_0}I_0\widetilde{V}_1,\widetilde{V}_2] \rangle\nonumber\\
& + \langle [ \Lxio I_0\widetilde{V}_2 + \Lxio I_0[ I_0\widetilde{V}_1, \langle V_1\rangle] + \frac{1}{2}\Lxio I_0[I_0\widetilde{V}_1,\widetilde{V}_1]^{\text{osc}} + \frac{1}{2}[\Lxio I_0\widetilde{V}_1, I_0\widetilde{V}_1],V_1] \rangle \nonumber\\
& -\mathcal{L}_{V_0}\bigg(\langle [\Lxio I_0\widetilde{V}_2,I_0\widetilde{V}_1]\rangle + \langle [\Lxio I_0[I_0\widetilde{V}_1,\langle V_1\rangle],I_0\widetilde{V}_1] \rangle\bigg)\nonumber\\
& -\mathcal{L}_{V_0}\bigg(\frac{1}{2}\langle [\Lxio I_0[I_0\widetilde{V}_1,\widetilde{V}_1]^{\text{osc}},I_0\widetilde{V}_1] \rangle + \frac{1}{6} \langle[[\Lxio I_0\widetilde{V}_1,I_0\widetilde{V}_1],I_0\widetilde{V}_1]\rangle\bigg)\nonumber\\
& = \langle [\mathcal{L}_{\xi_0}I_0\widetilde{V}_1,\widetilde{V}_2]\rangle + \frac{1}{2}\langle[[\mathcal{L}_{\xi_0}I_0\widetilde{V}_1,I_0\widetilde{V}_1],\langle V_1\rangle] \rangle + \frac{1}{2}\langle [[\mathcal{L}_{\xi_0}I_0\widetilde{V}_1,I_0\widetilde{V}_1],\widetilde{V}_1]\rangle\nonumber\\
& - \frac{1}{6}\langle [[\mathcal{L}_{\xi_0}\widetilde{V}_1,I_0\widetilde{V}_1],I_0\widetilde{V}_1]\rangle-\frac{1}{6}\langle [[\mathcal{L}_{\xi_0}I_0\widetilde{V}_1,\widetilde{V}_1],I_0\widetilde{V}_1] \rangle-\frac{1}{6}\langle [[\mathcal{L}_{\xi_0}I_0\widetilde{V}_1,I_0\widetilde{V}_1],\widetilde{V}_1] \rangle\nonumber\\
&+\left\langle \left[\bigg(\widetilde{V}_2 + [I_0\widetilde{V}_1,\langle V_1\rangle] + \frac{1}{2}[I_0\widetilde{V}_1,\widetilde{V}_1]^{\text{osc}}\bigg),\mathcal{L}_{\xi_0}I_0\widetilde{V}_1\right]\right\rangle\nonumber\\
& = \frac{1}{3}\langle [[\mathcal{L}_{\xi_0}I_0\widetilde{V}_1,I_0\widetilde{V}_1],\widetilde{V}_1]\rangle- \frac{1}{6}\langle [[\mathcal{L}_{\xi_0}\widetilde{V}_1,I_0\widetilde{V}_1],I_0\widetilde{V}_1]\rangle-\frac{1}{6}\langle [[\mathcal{L}_{\xi_0}I_0\widetilde{V}_1,\widetilde{V}_1],I_0\widetilde{V}_1] \rangle\nonumber\\
&+\frac{1}{2}\left\langle \left[[I_0\widetilde{V}_1,\widetilde{V}_1]^{\text{osc}},\mathcal{L}_{\xi_0}I_0\widetilde{V}_1\right]\right\rangle\nonumber\\
&=\frac{1}{3}\langle [[\mathcal{L}_{\xi_0}I_0\widetilde{V}_1,I_0\widetilde{V}_1],\widetilde{V}_1]\rangle - \frac{1}{3}\langle [[\widetilde{V}_1,I_0\widetilde{V}_1],\mathcal{L}_{\xi_0}I_0\widetilde{V}_1]\rangle+\frac{1}{3}\langle [[\widetilde{V}_1,\mathcal{L}_{\xi_0}I_0\widetilde{V}_1],I_0\widetilde{V}_1]\rangle\nonumber\\
& = 0.
\end{align}

\end{proof}

\section{Noether's theorem and adiabatic invariants\label{section_adiabatic}}
In the previous Section we explained that all nearly-periodic systems admit a roto-rate vector. In this sense every nearly-periodic system has an \emph{approximate} $U(1)$ symmetry. In this subsection we will show that if a nearly-periodic system happens to have a Hamiltonian structure as well then there is an approximate conserved quantity $\mu_\epsilon = \mu_0 + \epsilon\,\mu_1 + \epsilon^2\,\mu_2 + \dots$ associated with its approximate $U(1)$ symmetry.  In effect we will prove an asymptotic version of Noether's theorem that applies to Hamiltonian nearly-periodic systems. We will work in the setting of presymplectic Hamiltonian systems with $\epsilon$-dependent exact presymplectic structures. In the setting of $\epsilon$-independent exact symplectic Hamiltonian systems \cite{Kruskal_1962} gave an abstract proof of an analogous result, in the sense that formulas were not provided for the approximate conserved quantity. Here we will improve Kruskal's results by providing (the first several terms of) the missing formulas, and by allowing for a much broader class of nearly-periodic Hamiltonian systems. In particular we will provide formulas for $\mu_0,\mu_1,\mu_2,$ and $\mu_3$. For a discussion of some of the subtleties associated with adiabatic invariants for nearly-periodic Poisson systems, see \cite{Omohundro_1986}.

Before discussing the asymptotic version of Noether's theorem it is useful to discuss the usual Noether's theorem in a coordinate-independent manner. We will focus our attention on Noether's theorem for Hamiltonian systems on presymplectic manifolds that admit a $U(1)$ symmetry. 

To that end, suppose $X$ is a vector field on a manifold $Z$ and assume that there is a $1$-form $\vartheta$ and a smooth function $H$ such that $\iota_X\mathbf{d}\vartheta = -\mathbf{d}H$. The dynamical system defined by $X$ is then known as a (presymplectic) Hamiltonian system, the $2$-form $\omega= -\mathbf{d}\vartheta$ is called the presymplectic form, and the scalar $H$ is called the Hamiltonian. Noether's theorem applies to such systems. In particular if $\Phi_\theta:Z\rightarrow Z$ is a $U(1)$-action ($\theta\in U(1) = \mathbb{R}/2\pi$) on $Z$ that leaves the Hamiltonian invariant, $\Phi_\theta^*H = H$, and that leaves the presymplectic form invariant, $\Phi_\theta^*\omega = \omega$, then the scalar $\mu = \iota_\xi\langle\vartheta\rangle$ is a constant of motion for $X$. Here $\xi = \partial_\theta\Phi_\theta\mid_{\theta=0}$ is the infinitesimal generator for the $U(1)$-action and $\langle\vartheta\rangle = (2\pi)^{-1}\int_0^{2\pi}\Phi_\theta^*\vartheta\,d\theta.$ The scalar $\mu$ is the Noether-invariant associated with the $U(1)$-action $\Phi_\theta$. The proof that $\mu$ is a conserved quantity for $X$ follows from the following simple calculation: $\mathcal{L}_{X}\mu = \iota_X\mathbf{d}\iota_\xi\langle\vartheta\rangle = \iota_X\mathcal{L}_{\xi}\langle\vartheta\rangle - \iota_X\iota_\xi \mathbf{d}\langle\vartheta\rangle = \iota_X\mathcal{L}_{\xi}\langle\vartheta\rangle - \mathcal{L}_\xi H = 0$.

Now suppose that $\dot{z} = \epsilon^{-1}V_\epsilon(z)$ defines a nearly-periodic system that happens to be Hamiltonian. Concretely this means the following.
\begin{definition}[Nearly-periodic Hamiltonian system\label{ham_def}]
A nearly periodic system $\dot{z} = \epsilon^{-1}V_\epsilon(z)$ is a \emph{nearly-periodic Hamiltonian system} if there is some $1$-form $\vartheta_\epsilon$ and some function $H_\epsilon$ such that $\iota_{V_\epsilon}\mathbf{d}\vartheta_\epsilon = -\mathbf{d}H_\epsilon$. $H_\epsilon$ and $\vartheta_\epsilon$ are required to depend smoothly on $\epsilon$ in a neighborhood of $\epsilon = 0$.
\end{definition} 
\noindent By mimicking the key parts of the $U(1)$ Noether theorem from the previous paragraph we will now prove that there exists a formal power series $\mu_\epsilon = \mu_0 + \epsilon\,\mu_1 + \dots$ that is constant along integral curves of $V_\epsilon$ to all-orders in $\epsilon$. In other words $\mathcal{L}_{V_\epsilon}\mu_\epsilon = 0$ in the sense of formal power series. The proof will be consistent with this article's goal of avoiding the well-known coordinate transform-based methods.

Before giving the proof it is useful to first give a variant of a technical Lemma originally proved in \cite{Kruskal_1962}. (Kruskal refers to his ``Theorem of Phase Independence" in Section C.1.)
\begin{lemma}[Bootstrapping of $U(1)$ averages\label{bootstrap}]
Fix a nearly-periodic system $\dot{z} = \epsilon^{-1}V_\epsilon(x)$. Suppose that $\tau_\epsilon$ is any differential form on $Z$ that depends smoothly on $\epsilon$. If $\tau_\epsilon$ is constant along the flow of $V_\epsilon$, i.e. $\mathcal{L}_{V_\epsilon}\tau_\epsilon=0$, and is \emph{almost} $U(1)$-invariant in the sense that $\mathcal{L}_{\xi_0}\tau_0=0$, then in fact $\tau_\epsilon$ satisfies $\mathcal{L}_{\xi_\epsilon}\tau_\epsilon = 0$ to all orders in $\epsilon$.
\end{lemma}

\begin{proof}
As mentioned earlier in the proof of Theorem \ref{roto_rate_existence_thm}, \cite{Kruskal_1962} shows that there is a formal near-identity diffeomorphism $T_\epsilon:Z\rightarrow Z$ such that $T_{\epsilon *}\xi_\epsilon = \xi_0$. (In fact there are many such $T_\epsilon$.) Set $V_\epsilon^* = T_{\epsilon*}V_\epsilon$ and $\tau_\epsilon^* = T_{\epsilon *}\tau_\epsilon$. 

Since $\tau_\epsilon$ is constant along the $V_\epsilon$-flow it is also true that $\tau_\epsilon^*$ is constant along the $V_\epsilon^*$-flow, i.e. $\mathcal{L}_{V_\epsilon^*}\tau_\epsilon^*=0$. In light of the fact that $[\xi_0,V_\epsilon^*] = T_{\epsilon*}[\xi_\epsilon,V_\epsilon] = 0$ this implies $\mathcal{L}_{V_\epsilon^*}\mathcal{L}_{\xi_0}\tau_\epsilon^*=0$. The $O(1)$ coefficient of this formal power series identity is $\mathcal{L}_{V_0}\mathcal{L}_{\xi_0}\tau_0 = 0$, which is trivially satisfied because $\tau_\epsilon$ is nearly $U(1)$-invariant by hypothesis. On the other hand, the $O(\epsilon)$ coefficient is $\mathcal{L}_{V_0}\mathcal{L}_{\xi_0}\tau_1^* = 0$, which says that $\mathcal{L}_{\xi_0}\tau_1^*$ is constant along the $V_0$-flow. We claim that this can only be true if $\mathcal{L}_{\xi_0}\tau_1^*=0$. To see this set $\widetilde{\alpha} =\mathcal{L}_{\xi_0}\tau_1^* $. The Lie derivative of $\widetilde{\alpha}$ along $V_0$ is given by
\begin{align}
\mathcal{L}_{V_0}\widetilde{\alpha} &= \omega_0\iota_{\xi_0}\mathbf{d}\widetilde{\alpha} + \mathbf{d}(\omega_0\iota_{\xi_0}\widetilde{\alpha})\nonumber\\
& = \omega_0\mathcal{L}_{\xi_0}\widetilde{\alpha} + \mathbf{d}\omega_0\wedge \iota_{\xi_0}\widetilde{\alpha} = 0.\label{order_one_bootstrap}
\end{align}
Contracting this formula with $\xi_0$ therefore implies $\omega_0 \mathcal{L}_{\xi_0}\iota_{\xi_0}\widetilde{\alpha} =0$. Because the $U(1)$-average of $\widetilde{\alpha}$ is zero and $\omega_0$ is nowhere vanishing this requires $\iota_{\xi_0}\widetilde{\alpha}=0$. But by Eq.\,\eqref{order_one_bootstrap} this implies $\omega_0\mathcal{L}_{\xi_0}\widetilde{\alpha} = 0$, which can only be satisfied if $\widetilde{\alpha}=0$, as desired. This shows, in particular, that $\mathcal{L}_{\xi_0}\tau_\epsilon^* = O(\epsilon^2)$.

To complete the proof we will now show that if, for some integer $n\geq 2$, $\mathcal{L}_{\xi_0}\tau_\epsilon^* = O(\epsilon^n)$ then in fact $\mathcal{L}_{\xi_0}\tau_\epsilon^* = O(\epsilon^{n+1})$. If this is true then, by induction, $\mathcal{L}_{\xi_0}\tau_\epsilon^* =0$ as a formal power series, which would imply the desired result since $ T_\epsilon^*(\mathcal{L}_{\xi_0}\tau_\epsilon^*) = \mathcal{L}_{\xi_\epsilon}\tau_\epsilon$. Because $\mathcal{L}_{\xi_0}\tau_\epsilon^* = O(\epsilon^n)$ the differential forms $\mathcal{L}_{\xi_0}\tau_k^*$ for $k\in\{0,1,\dots, n-1\}$ must each vanish. Therefore $\mathcal{L}_{\xi_0}\tau_\epsilon^*=\epsilon^{n}\,\mathcal{L}_{\xi_0}\tau_n^* + O(\epsilon^{n+1})$. But since $\mathcal{L}_{V_\epsilon^*}\mathcal{L}_{\xi_0}\tau_\epsilon^* = 0$ to all orders in $\epsilon$ this means $\mathcal{L}_{V_0}\mathcal{L}_{\xi_0}\tau_n^* = 0$. Repeating the argument from the previous paragraph with $\widetilde{\alpha} = \mathcal{L}_{\xi_0}\tau_n^*$ then shows that in fact $\mathcal{L}_{\xi_0}\tau_n^*=0$. Therefore $\mathcal{L}_{\xi_0}\tau_\epsilon^*=\epsilon^{n}\,\mathcal{L}_{\xi_0}\tau_n^* + O(\epsilon^{n+1}) = O(\epsilon^{n+1})$, as claimed.

\end{proof}

Next we will show that the limiting roto-rate vector $\xi_0$ associated with a nearly-periodic Hamiltonian system is itself Hamiltonian.

\begin{lemma}[Hamiltonian structure of the Limiting roto-rate\label{energy_frequency}]
If $\dot{z} = \epsilon^{-1}V_\epsilon(z)$ is a nearly-periodic Hamiltonian system with frequency function $\omega_0$, limiting roto-rate $\xi_0$, presymplectic form $-\mathbf{d}\vartheta_\epsilon$, and Hamiltonian $H_\epsilon$ then there exists a function $\mu_0:Z\rightarrow \mathbb{R}$ such that $\omega_0^{-1}\mathbf{d}H_0 = \mathbf{d}\mu_0$. In particular, the limiting roto-rate $\xi_0$ satisfies $\iota_{\xi_0}\mathbf{d}\vartheta_0 = -\mathbf{d}\mu_0$, and is therefore Hamiltonian with presymplectic form $-\mathbf{d}\vartheta_0$ and Hamiltonian $\mu_0$.
\end{lemma}

\begin{proof}
Because $\iota_{V_\epsilon}\mathbf{d}\vartheta_\epsilon = -\mathbf{d}H_\epsilon$ and everything depends smoothly on $\epsilon$ it must also be true that $\omega_0\iota_{\xi_0}\mathbf{d}\vartheta_0 = -\mathbf{d}H_0$. Contracting both sides of this identity with $\xi_0$ implies $\mathcal{L}_{\xi_0}H_0 = 0$. Pulling back the identity along $\exp(\theta\,\xi_0)$ and then averaging over $\theta$ therefore implies $\omega_0\iota_{\xi_0}\mathbf{d}\langle \vartheta_0\rangle =-\mathbf{d}H_0 $. It follows that the function $\mu_0 = \iota_{\xi_0}\langle\vartheta_0 \rangle $ satisfies
\begin{align}
\mathbf{d}\mu_0 = \mathbf{d} \iota_{\xi_0}\langle\vartheta_0 \rangle = -\iota_{\xi_0}\mathbf{d}\langle \vartheta_0\rangle = \omega_0^{-1}\mathbf{d}H_0.
\end{align}
\end{proof}

Finally we will prove the existence of an adiabatic invariant for any nearly-periodic Hamiltonian system.

\begin{theorem}[existence of the adiabatic invariant\label{mu_exists_thm}]
Let $\dot{z} = \epsilon^{-1}V_\epsilon(z)$ be a Hamiltonian nearly-periodic system with presymplectic form $-\mathbf{d}\vartheta_\epsilon$ and Hamiltonian $H_\epsilon$. If $\xi_\epsilon$ denotes the associated roto-rate vector and $\overline{\vartheta}_\epsilon = (2\pi)^{-1}\int_0^{2\pi}\exp(\theta\,\xi_\epsilon)^*\vartheta_\epsilon\,d\theta$ denotes the formal $U(1)$ average of $\vartheta_\epsilon$ associated with $\xi_\epsilon$, then the formal power series
\begin{align}
\mu_\epsilon = \iota_{\xi_\epsilon}\overline{\vartheta}_\epsilon
\end{align}
satisfies $\mathcal{L}_{V_\epsilon}\mu_\epsilon = 0$. 

\end{theorem}

\begin{proof}
First observe that the $0$-form $H_\epsilon$ and the $2$-form $\mathbf{d}\vartheta_\epsilon$ are each constant along the flow $V_\epsilon$. Indeed, $\mathcal{L}_{V_\epsilon}H_\epsilon = \iota_{V_\epsilon}\mathbf{d}H_\epsilon =-\iota_{V_\epsilon}\iota_{V_\epsilon}\mathbf{d}\vartheta_\epsilon = 0$, and $\mathcal{L}_{V_\epsilon}\mathbf{d}\vartheta_\epsilon = \mathbf{d}\iota_{V_\epsilon}\mathbf{d}\vartheta_\epsilon = -\mathbf{d}\mathbf{d}H_\epsilon = 0$. (This is actually a general fact about presymplectic Hamiltonian systems.) Each of these forms is also nearly $U(1)$-invariant in the sense that $\mathcal{L}_{\xi_0}H_0 = 0$ and $\mathcal{L}_{\xi_0}\mathbf{d}\vartheta_0 = 0$. This can be seen by appealing to Lemma \ref{energy_frequency} and computing as follows:
\begin{align}
\mathcal{L}_{\xi_0}H_0& = \iota_{\xi_0}\mathbf{d}H_0 = \iota_{\xi_0}\omega_0\mathbf{d}\mu_0 = - \omega_0\iota_{\xi_0}\iota_{\xi_0}\mathbf{d}\vartheta_\epsilon = 0\\
\mathcal{L}_{\xi_0}\mathbf{d}\vartheta_0& = \mathbf{d}\iota_{\xi_0}\mathbf{d}\vartheta_0 = -\mathbf{d}\mathbf{d}\mu_0 = 0.
\end{align}
Therefore Lemma \ref{bootstrap} implies
\begin{align}
\mathcal{L}_{\xi_\epsilon}H_\epsilon& = 0\label{bootstrap_hamiltonian}\\
\mathcal{L}_{\xi_\epsilon}\mathbf{d}\vartheta_\epsilon & = 0,\label{bootstrap_two_form}
\end{align}
as formal power series.

It is now possible to directly compute $\mathcal{L}_{V_\epsilon}\mu_\epsilon$ using the formula
\begin{align}
\mathcal{L}_{V_\epsilon}\iota_{\xi_\epsilon}\overline{\vartheta}_\epsilon = \iota_{V_\epsilon}\mathbf{d}\iota_{\xi_\epsilon}\overline{\vartheta}_\epsilon =\iota_{V_\epsilon}\mathcal{L}_{\xi_\epsilon}\overline{\vartheta}_\epsilon + \iota_{\xi_\epsilon}\iota_{V_\epsilon}\mathbf{d}\overline{\vartheta}_\epsilon.
\end{align}
The first term on the right-hand-side vanishes to all orders in $\epsilon$ because 
\begin{align}
\mathcal{L}_{\xi_\epsilon}\overline{\vartheta}_\epsilon &= \frac{1}{2\pi}\int_0^{2\pi} \exp(\theta\,\xi_\epsilon)^*\mathcal{L}_{\xi_\epsilon}\vartheta_\epsilon\,d\theta\nonumber\\
& = \frac{1}{2\pi}\int_0^{2\pi} \frac{d}{d\theta}  \exp(\theta\,\xi_\epsilon)^*\vartheta_\epsilon\,d\theta\nonumber\\
& =  \exp(2\pi\,\xi_\epsilon)^*\vartheta_\epsilon -  \exp(0\,\xi_\epsilon)^*\vartheta_\epsilon\nonumber\\
&=\vartheta_\epsilon - \vartheta_\epsilon = 0.\label{bar_vartheta_is_averaged}
\end{align}
The second term on the right-hand-side also vanishes to all orders because, by Eq.\,\eqref{bootstrap_two_form}, 
\begin{align}
\mathbf{d}\overline{\vartheta}_\epsilon =\fint \exp(\theta\,\xi_\epsilon)^*\mathbf{d}\vartheta_\epsilon\,d\theta = \mathbf{d}\vartheta_\epsilon,\label{two_form_is_averaged}
\end{align}
which implies 
\begin{align}
\iota_{\xi_\epsilon}\iota_{V_\epsilon}\mathbf{d}\overline{\vartheta}_\epsilon = \iota_{\xi_\epsilon}\iota_{V_\epsilon}\mathbf{d}{\vartheta}_\epsilon = -\iota_{\xi_\epsilon}\mathbf{d}H_\epsilon = -\mathcal{L}_{\xi_\epsilon}H_\epsilon = 0,
\end{align}
by Eq.\,\eqref{bootstrap_hamiltonian}.
\end{proof}

According to this Theorem the quantity $\mu_\epsilon = \iota_{\xi_\epsilon}\overline{\vartheta}_\epsilon$ is an adiabatic invariant associated with any given nearly-periodic Hamiltonian system. In fact $\mu_\epsilon$ is equivalent to the adiabatic invariant discussed in \cite{Kruskal_1962} when the presymplectic form $-\mathbf{d}\vartheta_\epsilon$ is globally equivalent to the canonical symplectic form $\mathbf{d}q^i\wedge\mathbf{d}p_i$. (Note that no degenerate presymplectic form is even locally equivalent to the canonical symplectic form.) Therefore the first four coefficients of the expansion $\mu_\epsilon = \mu_0 + \epsilon\,\mu_1 + \epsilon^2\,\mu_2 + \dots$, expressed in terms of $V_\epsilon$, comprise the main objective of this Article. In principle the computation of these coefficients may be achieved using Theorem \ref{roto_formula_thm}, which gives explicit formulas for $\xi_0,\xi_1,\xi_2,\xi_3$ in terms of $V_\epsilon$. Indeed, with knowledge of $\xi_\epsilon$ the one-form $\overline{\vartheta}_\epsilon$ may be computed directly by expanding $\exp(\theta\,\xi_\epsilon)$ in its formal power series in $\epsilon$. Once $\overline{\vartheta}_\epsilon$ has been computed $\mu_\epsilon$ can be obtained by merely forming the contraction $\iota_{\xi_\epsilon}\overline{\vartheta}_\epsilon$.

The following Theorem and its proof performs such a calculation and records the resulting formulas for $\mu_0,\mu_1,\mu_2,\mu_3$. However, the proof does \emph{not} proceed along the direct route of first computing $\overline{\vartheta}_\epsilon$. Instead it employs a method that reveals a striking feature of the series $\mu_\epsilon = \iota_{\xi_\epsilon}\overline{\vartheta}_\epsilon$: if $\vartheta_\epsilon$ is subject to the gauge transformation $\vartheta_\epsilon\mapsto \vartheta_\epsilon + \alpha_\epsilon = \vartheta^\prime_\epsilon$, with $\alpha_\epsilon$ closed, then $\mu_\epsilon$ changes by at most a constant. This property may be seen abstractly using the following simple calculation,
\begin{align}
\mathbf{d}\iota_{\xi_\epsilon} \overline{\vartheta}_\epsilon^\prime & = \mathbf{d}\fint\iota_{\xi_\epsilon} \exp(\theta\,\xi_\epsilon)^*(\vartheta_\epsilon + \alpha_\epsilon)\,d\theta \nonumber\\
& = \fint \exp(\theta\,\xi_\epsilon)^*\mathbf{d}\iota_{\xi_\epsilon}(\vartheta_\epsilon + \alpha_\epsilon)\,d\theta\nonumber\\
& =  \fint \exp(\theta\,\xi_\epsilon)^*\mathcal{L}_{\xi_\epsilon}(\vartheta_\epsilon + \alpha_\epsilon)\,d\theta - \fint \exp(\theta\,\xi_\epsilon)^*\iota_{\xi_\epsilon}\mathbf{d}(\vartheta_\epsilon + \alpha_\epsilon)\,d\theta \nonumber\\
& = \fint \frac{d}{d\theta}\exp(\theta\,\xi_\epsilon)^*(\vartheta_\epsilon + \alpha_\epsilon)\,d\theta - \fint \exp(\theta\,\xi_\epsilon)^*\iota_{\xi_\epsilon}\mathbf{d}\vartheta_\epsilon\,d\theta\nonumber\\
& =  - \fint \exp(\theta\,\xi_\epsilon)^*\iota_{\xi_\epsilon}\mathbf{d}\vartheta_\epsilon\,d\theta,\label{indirect_homotopy}
\end{align}
which shows that $\mathbf{d}\mu_\epsilon$ is unchanged when $\vartheta_\epsilon$ is subject to a gauge transformation. In other words $\mathbf{d}\mu_\epsilon$ depends on $\vartheta_\epsilon$ only through the presymplectic form $-\mathbf{d}\vartheta_\epsilon$. The statement and proof of the following Theorem make this important property manifestly clear, order-by-order in $\epsilon$.

\begin{theorem}[Formulas for the adiabatic invariant\label{main_result}]
Suppose $\dot{z} = \epsilon^{-1}V_\epsilon(z)$ is a nearly-periodic Hamiltonian system with presymplectic form $-\mathbf{d}\vartheta_\epsilon$, Hamiltonian $H_\epsilon$, and limiting roto-rate vector $\xi_0$. The system's roto-rate vector $\xi_\epsilon$ is Hamiltonian in the sense that 
\begin{align}
\iota_{\xi_\epsilon}\mathbf{d}\vartheta_\epsilon = -\mathbf{d}\mu_\epsilon,\label{hamiltonian_roto}
\end{align} 
where $\mu_\epsilon$ is the system's adiabatic invariant. Moreover, the first four coefficients of the series $\mu_\epsilon = \mu_0 + \epsilon\,\mu_1 + \epsilon^2\,\mu_2 + \dots$ are given by 
\begin{align}
\mu_0 & =\iota_{\xi_0}\langle \vartheta_0\rangle  \label{mu0_formula}\\
\mu_1& = \iota_{\xi_0}\langle \vartheta_1\rangle - \mathcal{L}_{I_0\widetilde{V}_1}\mu_0\label{mu1_formula}\\
\mu_2&  =  \iota_{\xi_0}\langle \vartheta_2\rangle +\frac{1}{2}\left\langle \mathbf{d}\vartheta_0(\mathcal{L}_{\xi_0}I_0\widetilde{V}_1,I_0\widetilde{V}_1)\right\rangle+ \fint \bigg(   \frac{1}{2}\mathcal{L}_{Z_{1,\theta}}(\mu_1+\mu_1^\theta) +\mathcal{L}_{Z_{2,\theta}}\mu_0\bigg)\,d\theta\label{mu2_formula}\\
\mu_3& = \iota_{\xi_0}\langle \vartheta_3\rangle +\frac{2}{3}\langle \mathbf{d}\vartheta_2(I_0\widetilde{V}_1,\xi_0)\rangle -\frac{2}{3}\langle\mathbf{d}\vartheta_2\rangle(I_0\widetilde{V}_1,\xi_0) - \frac{1}{3}\langle \mathbf{d}\vartheta_1(\mathcal{L}_{\xi_0}I_0\widetilde{V}_1,I_0\widetilde{V}_1)\rangle\nonumber\\
& + \frac{1}{3}\left\langle \iota_{\mathcal{L}_{\xi_0}I_0\widetilde{V}_1}\mathbf{d}\vartheta_1\right\rangle(I_0\widetilde{V}_1) - \frac{1}{6}\langle \mathbf{d}\vartheta_0([\mathcal{L}_{\xi_0}I_0\widetilde{V}_1,I_0\widetilde{V}_1], I_0\widetilde{V}_1)\rangle \nonumber\\
&+ \frac{1}{6}\mathbf{d}\vartheta_0(\langle [\mathcal{L}_{\xi_0}I_0\widetilde{V}_1,I_0\widetilde{V}_1]\rangle,I_0\widetilde{V}_1) + \frac{1}{3}\langle \mathbf{d}\vartheta_0(\mathcal{L}_{\xi_0}I_0\widetilde{V}_1,I_0\widetilde{V}_2)\rangle\nonumber\\
& + \frac{1}{3}\langle \mathbf{d}\vartheta_0(\mathcal{L}_{\xi_0}I_0\widetilde{V}_1, I_0[I_0\widetilde{V}_1,\langle V_1\rangle])\rangle+\frac{1}{6}\langle \mathbf{d}\vartheta_0(\mathcal{L}_{\xi_0}I_0\widetilde{V}_1,I_0[I_0\widetilde{V}_1,\widetilde{V}_1]^{\text{osc}})\rangle\nonumber\\
& + \frac{1}{6}\langle  \langle\mathbf{d}\vartheta_1 \rangle(\mathcal{L}_{\xi_0}I_0\widetilde{V}_1,I_0\widetilde{V}_1)\rangle - \frac{1}{6}\mathcal{L}_{I_0\widetilde{V}_1}\langle \mathbf{d}\vartheta_0(\mathcal{L}_{\xi_0}I_0\widetilde{V}_1,I_0\widetilde{V}_1) \rangle\nonumber\\
&+\fint\bigg(\frac{2}{3}\mathcal{L}_{Z_{2,\theta}}\mu_1 + \frac{1}{3}\mathcal{L}_{Z_{2,\theta}}\mu_1^\theta + \frac{1}{6}\mathcal{L}^2_{Z_{1,\theta}}\mu_1^\theta + \frac{1}{3}\mathcal{L}_{Z_{1,\theta}}\mu_2\bigg)\,d\theta\nonumber\\
& + \fint \mathcal{L}_{Z_{3,\theta} + \frac{1}{6}[Z_{1,\theta},Z_{2,\theta}]}\mu_0\,d\theta \label{mu3_formula}
\end{align}
where the vector fields $Z_{1,\theta},Z_{2,\theta},Z_{3,\theta}$ are given by
\begin{align}
Z_{1,\theta} & = I_0\{\widetilde{V}_1\}\label{Z1_thm}\\
Z_{2,\theta} & =   I_0\{\widetilde{V}_2\} +   I_0[ I_0\{\widetilde{V}_1\}, \langle V_1\rangle] + \frac{1}{2}  I_0\{[I_0\widetilde{V}_1^{},\widetilde{V}_1^{}]\}^{\text{osc}} - \frac{1}{2} [I_0\widetilde{V}_1, I_0\widetilde{V}_1^{\theta}]\label{Z2_thm}\\
Z_{3,\theta} & = - \frac{1}{2}\int_0^\theta I_0[\langle [\Lxio I_0\widetilde{V}_1,I_0\widetilde{V}_1]\rangle,\widetilde{V}_1^{\theta_1}]\,d\theta_1\nonumber\\
& +I_0\{\widetilde{V}_3\} +  I_0[ I_0\{\widetilde{V}_1\},\langle V_2\rangle]^{\text{osc}} +I_0[ I_0\{\widetilde{V}_2\},\langle V_1\rangle]^{\text{osc}}\nonumber\\
&+\frac{1}{2} I_0\{[I_0\widetilde{V}_2^{},\widetilde{V}_1^{}]\}^{\text{osc}}+ \frac{1}{2}I_0\{[I_0\widetilde{V}_1^{},\widetilde{V}_2^{}]\}^{\text{osc}}+\frac{1}{3}I_0\{[I_0\widetilde{V}_1^{},[I_0\widetilde{V}_1^{},\widetilde{V}_1^{}]]\}^{\text{osc}}\nonumber\\
&+I_0[  I_0[ I_0\{\widetilde{V}_1\}, \langle V_1\rangle] ,\langle V_1\rangle]
+ \frac{1}{2}I_0[    I_0\{[I_0\widetilde{V}_1^{},\widetilde{V}_1^{}]\}^{\text{osc}} ,\langle V_1\rangle]\nonumber\\
 &+\frac{1}{2}I_0\{[I_0\widetilde{V}_1^{},[I_0\widetilde{V}_1^{},\langle V_1\rangle]]\}^{\text{osc}}+ \frac{1}{12} [I_0(\widetilde{V}_1^{\theta}+\widetilde{V}_1),[I_0\widetilde{V}_1,I_0\widetilde{V}_1^{\theta}]]+\frac{1}{2}\{[I_0\widetilde{V}_1,I_0\widetilde{V}_2]\}\nonumber\\
 &+\frac{1}{2} \bigg[  I_0\{\widetilde{V}_2\} +   I_0[ I_0\{\widetilde{V}_1\}, \langle V_1\rangle] + \frac{1}{2}  I_0\{[I_0\widetilde{V}_1,\widetilde{V}_1]\}^{\text{osc}},I_0(\widetilde{V}_1 + \widetilde{V}_1^\theta)\bigg],\label{Z3_thm}
\end{align}
where $\{A\} = A^\theta - A$ for any vector field $A$.
\end{theorem}

\begin{proof}
In order to establish Eq.\,\eqref{hamiltonian_roto} it is sufficient to compute the exterior derivative of $\mu_\epsilon = \iota_{\xi_\epsilon}\overline{\vartheta}_\epsilon$ directly:
\begin{align}
\mathbf{d}\mu_\epsilon  = \mathbf{d}\iota_{\xi_\epsilon}\overline{\vartheta}_\epsilon = \mathcal{L}_{\xi_\epsilon}\overline{\vartheta}_\epsilon - \iota_{\xi_{\epsilon}}\mathbf{d}\overline{\vartheta}_\epsilon = -\iota_{\xi_{\epsilon}}\mathbf{d}\vartheta_\epsilon,
\end{align}
where we have made use of Eqs.\,\eqref{bar_vartheta_is_averaged} and \eqref{two_form_is_averaged} from the proof of Lemma \ref{mu_exists_thm}.

In order to obtain formulas for the $\mu_k$ we will proceed in two steps. First we will use Stokes' theorem to identify an alternative all-orders expression for $\mu_\epsilon$ that obviates how $\mu_\epsilon$ changes when $\vartheta_\epsilon$ is subject to the gauge transformation $\vartheta_\epsilon \mapsto \vartheta_\epsilon + \alpha_\epsilon$ with $\alpha_\epsilon$ closed. Then we will use the perturbative BCH formula (c.f. Lemma \ref{BCH})  to expand the resulting expression as a power series in $\epsilon$.

Fix $z\in Z$ and define the mapping $S:S^1\times [0,\epsilon]\rightarrow Z:(\theta,\lambda)\mapsto \exp(\theta\,\xi_{\lambda})(z)$. Choose an orientation for $S^1\times[0,\epsilon]$ by declaring that the ordered basis $(\partial_\theta,\partial_{\lambda})$ is positively oriented. By Stokes' theorem
\begin{align}
\int_{S^1\times[0,\epsilon]}\mathbf{d}S^*\vartheta_\epsilon = \int_{S^1\times\{0\}}S^*\vartheta_\epsilon - \int_{S^1\times\{\epsilon\}}S^*\vartheta_\epsilon,\label{stokes_abstract}
\end{align}
where $S^1\times\{0\}$ and $S^1\times\{\epsilon\}$ are each oriented in the sense of increasing $\theta\in S^1$. Accounting for these orientation conventions, Eq.\,\eqref{stokes_abstract} may be re-written in terms of definite integrals as
\begin{align}
\int_0^{\epsilon}\int_0^{2\pi} [S^*\mathbf{d}\vartheta_\epsilon](\partial_\theta,\partial_{\lambda})\,d\theta\,d\lambda& = \int_0^{2\pi} [\vartheta_\epsilon(\xi_0)](\exp(\theta\,\xi_0)(z)) \,d\theta - \int_0^{2\pi} [\vartheta_\epsilon(\xi_\epsilon)](\exp(\theta\,\xi_\epsilon)(z))\,d\theta\nonumber\\
& = 2\pi \iota_{\xi_0}\langle \vartheta_\epsilon\rangle(z) - 2\pi\iota_{\xi_\epsilon}\overline{\vartheta}_\epsilon(z),
\end{align}
which shows that the adiabatic invariant $\mu_\epsilon = \iota_{\xi_\epsilon}\overline{\vartheta}_\epsilon$ may be expressed as
\begin{align}
\mu_\epsilon(z) = \iota_{\xi_0}\langle \vartheta_\epsilon\rangle(z)+ \fint \int_0^{\epsilon}[S^*\mathbf{d}\vartheta_\epsilon](\partial_{\lambda},\partial_\theta)\,d\theta\,d\lambda.\label{proto_mu_formula}
\end{align}
Note that the second term on the right-hand-side is in a somewhat unwieldy form. To rectify this issue first observe that the partial derivatives of $\exp(\theta\,\xi_{\lambda})$ may be expressed as
\begin{align}
\partial_\theta \exp(\theta\,\xi_{\lambda}) &= \xi_{\lambda}\circ \exp(\theta\,\xi_{\lambda})\\
\partial_{\lambda}\exp(\theta\,\xi_{\lambda}) & = \bigg(\exp(-\theta\,\xi_0)^*\phi(-\mathcal{L}_{Z_{\lambda,\theta}})\partial_{\lambda}Z_{\lambda,\theta}\bigg)\circ \exp(\theta\,\xi_{\lambda}),
\end{align}
where $Z_{\lambda,\theta} = \text{ln}(\exp(-\theta\,\xi_0)\circ \exp(\theta\,\xi_{\lambda}))$, $\phi(z) = (\exp(z) - 1)/z$, and we have made use of  Eq.\,\eqref{useful_derivaitive_formula_exp}. Therefore the scalar $[S^*\mathbf{d}\vartheta_\epsilon](\partial_{\lambda},\partial_\theta)$ may be written
\begin{align}
[S^*\mathbf{d}\vartheta_\epsilon](\partial_{\lambda},\partial_\theta) &=\exp(\theta\,\xi_{\lambda})^*\bigg(\mathbf{d}\vartheta_\epsilon\bigg(\exp(-\theta\,\xi_0)^*\phi(-\mathcal{L}_{Z_{\lambda,\theta}})\partial_{\lambda}Z_{\lambda,\theta},\xi_{\lambda}\bigg) \bigg)(z)\nonumber\\
 &=[\exp(Z_{\lambda,\theta})^*\mathbf{d}\vartheta_\epsilon^\theta]\bigg(\exp(Z_{\lambda,\theta})^*\phi(-\mathcal{L}_{Z_{\lambda,\theta}})\partial_{\lambda}Z_{\lambda,\theta},\xi_{\lambda}\bigg) (z)\nonumber\\
 &=[\exp(Z_{\lambda,\theta})^*\mathbf{d}\vartheta_\epsilon^\theta]\bigg(\phi(\mathcal{L}_{Z_{\lambda,\theta}})\partial_{\lambda}Z_{\lambda,\theta},\xi_{\lambda}\bigg) (z).
\end{align}
An all-orders formula for the adiabatic invariant $\mu_\epsilon$ is therefore
\begin{align}
\mu_\epsilon = \iota_{\xi_0}\langle \vartheta_\epsilon\rangle + \fint\int_0^\epsilon [\exp(Z_{\lambda,\theta})^*\mathbf{d}\vartheta_\epsilon^\theta]\bigg(\phi(\mathcal{L}_{Z_{\lambda,\theta}})\partial_{\lambda}Z_{\lambda,\theta},\xi_{\lambda}\bigg)\,d\lambda\,d\theta.\label{mu_all_orders}
\end{align}
This is the formula we will use to compute the coefficients of the series $\mu_\epsilon$. As promised, if $\vartheta_\epsilon$ is subject to the gauge transformation $\vartheta_\epsilon\mapsto \vartheta_\epsilon + \alpha_\epsilon$, with $\alpha_\epsilon$ closed, the formula \eqref{mu_all_orders} shows that $\mu_\epsilon$ transforms as $\mu_\epsilon\mapsto \mu_\epsilon + \iota_{\xi_0}\langle \alpha_\epsilon\rangle$. The change in $\mu_\epsilon$, $\Delta\mu_\epsilon = \iota_{\xi_0}\langle \alpha_\epsilon\rangle$, evaluated at $z\in Z$ may therefore be written as the closed loop integral
\begin{align}
\Delta\mu_\epsilon(z) = \frac{1}{2\pi}\oint_{\gamma_z} \alpha_\epsilon,
\end{align}
where the $z$-dependent curve $\gamma_z$ is given by $\gamma_z(\theta) = \exp(\theta\,\xi_0)(z)$. Because $\alpha_\epsilon$ is closed the integral $\oint_{\gamma_z}\alpha_\epsilon$ depends only on the homotopy class of $\gamma_z$. Since $\gamma_z$ depends continuously on $z$ this means $\Delta\mu_\epsilon$ is constant on path-connected components of $Z$. In particular, because the path-components are open subsets of $Z$,  $\mathbf{d}\Delta\mu_\epsilon = 0$, as inferred earlier from the abstract argument related to Eq.\,\eqref{indirect_homotopy}. In fact, this argument shows that the change in $\mu_\epsilon$ induced by a gauge transformation is equal to ($1/(2\pi)$ times) the cohomology class of $\alpha_\epsilon$ paired with the cycle defined by the $\xi_0$-orbit $\gamma_z$. In particular if either (a) the first deRham cohomology group of $Z$ is trivial, or (b) the $\xi_0$-orbits are each homologous to a point then $\mu_\epsilon$ is gauge-independent.

Expanding Eq.\,\eqref{mu_all_orders} in a power series is facilitated by first recording the more elementary power series expansions involving the coefficients of $Z_{\lambda,\theta} = \epsilon\,Z_{1,\theta} + \epsilon^2\,Z_{2,\theta} + \epsilon^3\,Z_{3,\theta} +\dots$ given by
\begin{align}
\exp(Z_{\lambda,\theta})^*\mathbf{d}\vartheta_\epsilon^\theta &= \mathbf{d}\vartheta_\epsilon^\theta + \lambda \mathcal{L}_{Z_{1,\theta}}\mathbf{d}\vartheta_\epsilon^\theta +\lambda^2 \bigg(\mathcal{L}_{Z_{2,\theta}} + \frac{1}{2}\mathcal{L}_{Z_{1,\theta}}^2\bigg)\mathbf{d}\vartheta_\epsilon^\theta+O(\lambda^3) \\
\phi(\mathcal{L}_{Z_{\lambda,\theta}})\partial_\lambda Z_{\lambda,\theta} & = Z_{1,\theta} + \lambda\,2Z_{2,\theta} + \lambda^2 \bigg(3 Z_{3,\theta} + \frac{1}{2}[Z_{1,\theta},Z_{2,\theta}]\bigg)+O(\lambda^3).
\end{align}
Substituting these expansion into Eq.\,\eqref{mu_all_orders} and performing the $\lambda$-integrals then gives
\begin{align}
\mu_\epsilon & = \iota_{\xi_0}\langle\vartheta_\epsilon \rangle  + \epsilon\fint \mathbf{d}\vartheta_\epsilon^\theta (Z_{1,\theta},\xi_0)\,d\theta \nonumber\\
&+ \epsilon^2\fint\bigg(\frac{1}{2}\left[\mathcal{L}_{Z_{1,\theta}}\mathbf{d}\vartheta_\epsilon^\theta\right] (Z_{1,\theta},\xi_0) + \frac{1}{2} \mathbf{d}\vartheta_\epsilon^\theta(Z_{1,\theta},\xi_1)+  \mathbf{d}\vartheta_\epsilon^\theta(Z_{2,\theta},\xi_0)\bigg)\,d\theta\nonumber\\
& + \epsilon^3\fint\bigg(\frac{2}{3}  \left[\mathcal{L}_{Z_{1,\theta}}\mathbf{d}\vartheta_\epsilon^\theta\right] (Z_{2,\theta},\xi_0) + \frac{2}{3} \mathbf{d}\vartheta_\epsilon^\theta (Z_{2,\theta},\xi_1) + \frac{1}{3} \left[\mathcal{L}_{Z_{1,\theta}}\mathbf{d}\vartheta_\epsilon^\theta\right](Z_{1,\theta},\xi_1)\nonumber\\
& \quad\quad \quad+ \frac{1}{3}\left[\mathcal{L}_{Z_{2,\theta}}\mathbf{d}\vartheta_\epsilon^\theta + \frac{1}{2}\mathcal{L}_{Z_{1,\theta}}^2\mathbf{d}\vartheta_\epsilon^\theta\right](Z_{1,\theta},\xi_0) + \mathbf{d}\vartheta_\epsilon^\theta\left(Z_{3,\theta} + \frac{1}{6}[Z_{1,\theta},Z_{2,\theta}],\xi_0\right)\nonumber\\
& \quad \quad \quad + \frac{1}{3}\mathbf{d}\vartheta_\epsilon^\theta(Z_{1,\theta},\xi_2)\bigg)\,d\theta + O(\epsilon^4).\label{mu_setup}
\end{align}
The task of finding formulas for the $\mu_k$ is therefore reduced to the problem of finding expressions for the $Z_{k,\theta}$ and then substituting them into Eq.\,\eqref{mu_setup}.

In order to compute terms in the series $Z_{\lambda,\theta}$ it is helpful to reuse the perturbative BCH formula provided by Lemma \ref{BCH}. Setting $ \epsilon = \lambda$, $A = \theta\,\xi_0$, and $B = \theta (\xi_1 + \lambda\,\xi_2 +\lambda^2 \xi_3 + \dots)$ the first several coefficients of $Z_{\lambda,\theta}$ given by Lemma \ref{BCH} may be expressed in integral form as
\begin{align}
Z_{0,\theta} & = 0\label{Z0_formula_one_form_thm}\\
Z_{1,\theta} &= \int_0^\theta \xi_1^{\theta_1}\,d\theta_1\label{Z1_formula_one_form_thm}\\
Z_{2,\theta} & =\int_0^\theta \xi_2^{\theta_1}\,d\theta_1+ \frac{1}{2}\int_0^\theta\int_0^{\theta_1}[\xi_1^{\theta_2},\xi_1^{\theta_1}]\,d\theta_2\,d\theta_1 \label{Z2_formula_one_form_thm}\\
Z_{3,\theta} & = \int_0^\theta \xi_3^{\theta_1}\,d\theta_1 + \frac{1}{2}\int_0^\theta\int_0^{\theta_1}[\xi_1^{\theta_2},\xi_2^{\theta_1}]\,d\theta_2\,d\theta_1+  \frac{1}{2}\int_0^\theta\int_0^{\theta_1}[\xi_2^{\theta_2},\xi_1^{\theta_1}]\,d\theta_2\,d\theta_1\nonumber\\
& + \frac{1}{6}\int_0^\theta\int_0^{\theta_1}\int_0^{\theta_2}\bigg([\xi_1^{\theta_3},[\xi_1^{\theta_2},\xi_1^{\theta_1}]] + [[\xi_1^{\theta_3},\xi_1^{\theta_2}],\xi_1^{\theta_1}]\bigg)\,d\theta_3\,d\theta_2\,d\theta_1.
\end{align}
Remarkably, most of the integrations indicated here may be carried out explicitly. First consider $Z_{1,\theta}$. By inserting Eq.\,\eqref{xi1_formula} for $\xi_1$ into Eq.\,\eqref{Z1_formula_one_form_thm} the vector field $Z_{1,\theta}$ may be expressed as
\begin{align}
Z_{1,\theta} &=  \int_0^\theta \Lxio I_0\widetilde{V}_1^{\theta_1}\,d\theta_1\nonumber\\
&= I_0\widetilde{V}_1^{\theta} - I_0\widetilde{V}_1^{},\label{Z1_explicit}
\end{align}
which no longer contains any integrals. Similarly, by inserting Eqs.\,\eqref{xi1_formula} and \eqref{xi2_formula} into Eq.\,\eqref{Z2_formula_one_form_thm} the vector field $Z_{2,\theta}$ becomes
\begin{align}
Z_{2,\theta} & = \int_0^\theta \bigg( \Lxio I_0\widetilde{V}_2^{\theta_1} + \Lxio I_0[ I_0\widetilde{V}_1^{\theta_1}, \langle V_1\rangle] + \frac{1}{2}\Lxio I_0[I_0\widetilde{V}_1^{\theta_1},\widetilde{V}_1^{\theta_1}]^{\text{osc}}  \bigg)\,d\theta_1\nonumber\\
& +\frac{1}{2} \int_0^\theta [\Lxio I_0\widetilde{V}_1^{\theta_1}, I_0\widetilde{V}_1^{\theta_1}]\,d\theta_1 + \frac{1}{2}\int_0^\theta\int_0^{\theta_1}[\Lxio I_0\widetilde{V}_1^{\theta_2},\Lxio I_0\widetilde{V}_1^{\theta_1}]\,d\theta_2\,d\theta_1\nonumber\\
& =   I_0\widetilde{V}_2^{\theta} +   I_0[ I_0\widetilde{V}_1^{\theta}, \langle V_1\rangle] + \frac{1}{2}  I_0[I_0\widetilde{V}_1^{\theta},\widetilde{V}_1^{\theta}]^{\text{osc}} \nonumber\\
& -I_0\widetilde{V}_2^{} -  I_0[ I_0\widetilde{V}_1^{}, \langle V_1\rangle] - \frac{1}{2}  I_0[I_0\widetilde{V}_1^{},\widetilde{V}_1^{}]^{\text{osc}}  +\frac{1}{2} \int_0^\theta [\Lxio I_0\widetilde{V}_1^{\theta_1}, I_0\widetilde{V}_1^{\theta_1}]\,d\theta_1\nonumber\\
&+\frac{1}{2}\int_0^\theta[I_0\widetilde{V}_1^{\theta_1},\Lxio I_0\widetilde{V}_1^{\theta_1}]\,d\theta_1 - \frac{1}{2}\int_0^\theta [I_0\widetilde{V}_1,\Lxio I_0\widetilde{V}_1^{\theta_1}]\,d\theta_1\nonumber\\
& =   I_0\widetilde{V}_2^{\theta} +   I_0[ I_0\widetilde{V}_1^{\theta}, \langle V_1\rangle] + \frac{1}{2}  I_0[I_0\widetilde{V}_1^{\theta},\widetilde{V}_1^{\theta}]^{\text{osc}} \nonumber\\
& -I_0\widetilde{V}_2^{} -   I_0[ I_0\widetilde{V}_1^{}, \langle V_1\rangle] - \frac{1}{2}  I_0[I_0\widetilde{V}_1^{},\widetilde{V}_1^{}]^{\text{osc}}   - \frac{1}{2} [I_0\widetilde{V}_1, I_0\widetilde{V}_1^{\theta}],\label{Z2_explicit}
\end{align}
where the integrals that could not be evaluated by recognizing a total derivative have cancelled, apparently fortuitously. Finally consider $Z_{3,\theta}$, which is the sum of a single integral, a pair of double integrals, and a triple integral. The triple integral may be partially simplified by making use of Eq.\,\eqref{xi1_formula} for $\xi_1$ and recognizing total derivatives according to
\begin{align}
&\frac{1}{6}\int_0^\theta\int_0^{\theta_1}\int_0^{\theta_2}\bigg([\xi_1^{\theta_3},[\xi_1^{\theta_2},\xi_1^{\theta_1}]] + [[\xi_1^{\theta_3},\xi_1^{\theta_2}],\xi_1^{\theta_1}]\bigg)\,d\theta_3\,d\theta_2\,d\theta_1 \nonumber\\ 
=&\frac{1}{3}\int_0^\theta[[\Lxio I_0\widetilde{V}_1^{\theta_2},I_0\widetilde{V}_1^{\theta_2}],I_0\widetilde{V}_1^{\theta_2}]\,d\theta_2\nonumber\\
+&\frac{1}{6}\int_0^\theta\bigg(\frac{1}{2}\partial_{\theta_2}[I_0\widetilde{V}_1^{\theta_2},[I_0\widetilde{V}_1^{\theta_2},I_0(\widetilde{V}_1^{\theta} + \widetilde{V}_1)]] - \frac{3}{2}[[\Lxio I_0\widetilde{V}_1^{\theta_2},I_0\widetilde{V}_1^{\theta_2}],I_0(\widetilde{V}_1^\theta + \widetilde{V}_1)] \bigg)\,d\theta_2\nonumber\\
+&\frac{1}{6}[I_0\widetilde{V}_1^{},[I_0\widetilde{V}_1^{},I_0\widetilde{V}_1^{\theta}]] - \frac{1}{6}[[I_0\widetilde{V}_1,I_0\widetilde{V}_1^\theta],I_0\widetilde{V}_1^\theta] \nonumber\\
=&\frac{1}{3}\int_0^\theta[[\Lxio I_0\widetilde{V}_1^{\theta_2},I_0\widetilde{V}_1^{\theta_2}],I_0\widetilde{V}_1^{\theta_2}]\,d\theta_2\nonumber\\
-&\frac{1}{4}\int_0^\theta[[\Lxio I_0\widetilde{V}_1^{\theta_2},I_0\widetilde{V}_1^{\theta_2}],I_0(\widetilde{V}_1^\theta + \widetilde{V}_1)] \,d\theta_2\nonumber\\
+& \frac{1}{12} [I_0(\widetilde{V}_1^{\theta}+\widetilde{V}_1),[I_0\widetilde{V}_1,I_0\widetilde{V}_1^{\theta}]],\label{triple_integral}
\end{align}
where we have used the identity
\begin{align}
[I_0\widetilde{V}_1^{\theta_2},[\Lxio I_0\widetilde{V}_1^{\theta_2},I_0(\widetilde{V}_1^{\theta} + \widetilde{V}_1)]]& = \frac{1}{2}\partial_{\theta_2}[I_0\widetilde{V}_1^{\theta_2},[I_0\widetilde{V}_1^{\theta_2},I_0(\widetilde{V}_1^{\theta} + \widetilde{V}_1)]]\nonumber\\
&-\frac{1}{2}[[\Lxio I_0\widetilde{V}_1^{\theta_2},I_0\widetilde{V}_1^{\theta_2}],I_0(\widetilde{V}_1^{\theta}+\widetilde{V}_1)].
\end{align}
The double integrals may be partially simplified in a similar manner upon making use of Eqs.\,\eqref{xi1_formula} and \eqref{xi2_formula} according to 
\begin{align}
 &\frac{1}{2}\int_0^\theta\int_0^{\theta_1}[\xi_1^{\theta_2},\xi_2^{\theta_1}]\,d\theta_2\,d\theta_1+  \frac{1}{2}\int_0^\theta\int_0^{\theta_1}[\xi_2^{\theta_2},\xi_1^{\theta_1}]\,d\theta_2\,d\theta_1\nonumber\\
 +& \frac{1}{2}\int_0^\theta[\xi_2^{\theta_2},I_0\widetilde{V}_1^{\theta}]\,d\theta_2 - \frac{1}{2}\int_0^\theta[\xi_2^{\theta_2},I_0\widetilde{V}_1^{\theta_2}]\,d\theta_2\nonumber\\
 =&\int_0^\theta[I_0\widetilde{V}_1^{\theta_1},\xi_2^{\theta_1}]\,d\theta_1+\frac{1}{2}\int_0^\theta[\xi_2^{\theta_1},I_0(\widetilde{V}_1^{}+\widetilde{V}_1^\theta)]\,d\theta_1\nonumber\\
 =&\frac{1}{2} \bigg[  I_0\widetilde{V}_2^\theta +   I_0[ I_0\widetilde{V}_1^\theta, \langle V_1\rangle] + \frac{1}{2}  I_0[I_0\widetilde{V}_1^\theta,\widetilde{V}_1^\theta]^{\text{osc}},I_0(\widetilde{V}_1 + \widetilde{V}_1^\theta)\bigg]+\frac{1}{2}[I_0\widetilde{V}_1^\theta,I_0\widetilde{V}_2^\theta]\nonumber\\
 -&\frac{1}{2} \bigg[  I_0\widetilde{V}_2 +   I_0[ I_0\widetilde{V}_1, \langle V_1\rangle] + \frac{1}{2}  I_0[I_0\widetilde{V}_1,\widetilde{V}_1]^{\text{osc}},I_0(\widetilde{V}_1 + \widetilde{V}_1^\theta)\bigg] - \frac{1}{2}[I_0\widetilde{V}_1,I_0\widetilde{V}_2]\nonumber\\
 +&\frac{1}{4}\int_0^\theta[[\Lxio I_0\widetilde{V}_1^{\theta_1}, I_0\widetilde{V}_1^{\theta_1}],I_0(\widetilde{V}_1+\widetilde{V}_1^\theta)]\,d\theta_1\nonumber\\
 +&\frac{1}{2}\int_0^\theta[I_0\widetilde{V}_1^{\theta_1},\Lxio I_0\widetilde{V}_2^{\theta_1}]\,d\theta_1+\frac{1}{2}\int_0^\theta[I_0\widetilde{V}_2^{\theta_1},\Lxio I_0\widetilde{V}_1^{\theta_1}]\,d\theta_1 \nonumber\\
 +&  \frac{1}{2}\int_0^\theta[I_0\widetilde{V}_1^{\theta_1},\langle [\Lxio I_0\widetilde{V}_1^{}, I_0\widetilde{V}_1^{}]\rangle]\,d\theta_1+\int_0^\theta[I_0\widetilde{V}_1^{\theta_1}, I_0[  \Lxio I_0\widetilde{V}_1^{\theta_1}, V^{\theta_1}]^{\text{osc}}]\,d\theta_1 .\label{double_integral}
\end{align}
Adding expressions \eqref{double_integral} and \eqref{triple_integral} to $\int_0^\theta\xi_3^{\theta_1}\,d\theta_1$ with $\xi_3$ given by \eqref{xi3_formula}, and accounting for the various fortuitous cancellations that occur, the net result for $Z_{3,\theta}$ is Eq.\,\eqref{Z3_thm}.

Finally, we can substitute Eqs.\,\eqref{Z1_thm},\eqref{Z2_thm}, and \eqref{Z3_thm} into the formula \eqref{mu_setup} in order to obtain explicit expressions for $\mu_0,\mu_1,\mu_2,$ and $\mu_3$.
The $O(1)$ terms in Eq.\,\eqref{mu_setup} give the result $\mu_0 = \iota_{\xi_0}\langle \vartheta_0\rangle$, which is consistent with Lemma \ref{energy_frequency} and Eq.\,\eqref{mu0_formula}. Note that Lemma \ref{energy_frequency} says $\mathbf{d}\mu_0 = \omega_0^{-1}\mathbf{d}H_0$, which generalizes the commonly-encountered expression that gives the adiabatic invariant as $(\text{energy})/(\text{frequency})$. The $O(\epsilon)$ terms in Eq.\,\eqref{mu_setup} give the result
\begin{align}
\mu_1 & = \iota_{\xi_0}\langle \vartheta_1\rangle  + \fint \mathbf{d}\vartheta_0(Z_{1,\theta},\xi_0)\,d\theta\nonumber\\
& =  \iota_{\xi_0}\langle \vartheta_1\rangle + \fint \mathcal{L}_{Z_{1,\theta}}\mu_0\,d\theta\nonumber\\
& = \iota_{\xi_0}\langle \vartheta_1\rangle - \mathcal{L}_{I_0\widetilde{V}_1}\mu_0,
\end{align}
which reproduces Eq.\,\eqref{mu1_formula}. The $O(\epsilon^2)$ terms of Eq.\,\eqref{mu_setup} give
\begin{align*}
\mu_2 &= \iota_{\xi_0}\langle \vartheta_2\rangle +\fint \mathbf{d}\vartheta_1^\theta(Z_{1,\theta},\xi_0)\,d\theta\nonumber\\
& + \fint \bigg( \frac{1}{2}\left[\mathcal{L}_{Z_{1,\theta}}\mathbf{d}\vartheta_0\right](Z_{1,\theta},\xi_0) + \frac{1}{2}\mathbf{d}\vartheta_0(Z_{1,\theta},\xi_1) + \mathbf{d}\vartheta_0(Z_{2,\theta},\xi_0)\bigg)\,d\theta,
\end{align*}
which may be simplified by making use of the identities
\begin{gather}
\mathbf{d}\vartheta_1^\theta + \mathcal{L}_{Z_{1,\theta}}\mathbf{d}\vartheta_0 = \mathbf{d}\vartheta_1\label{first_order_presymplectic_symmetry}\\
\iota_{\xi_1}\mathbf{d}\vartheta_0 + \iota_{\xi_0}\mathbf{d}\vartheta_1 = -\mathbf{d}\mu_1.\label{first_order_hamiltonian_mu}
\end{gather}
In particular,
\begin{align}
\mu_2 & =  \iota_{\xi_0}\langle \vartheta_2\rangle +\fint \mathbf{d}\vartheta_1^\theta(Z_{1,\theta},\xi_0)\,d\theta\nonumber\\
& + \fint \bigg( \frac{1}{2}\left[\mathbf{d}\vartheta_1 - \mathbf{d}\vartheta_1^\theta\right](Z_{1,\theta},\xi_0) + \frac{1}{2}(\mathbf{d}\mu_1 + \iota_{\xi_0}\mathbf{d}\vartheta_1)(Z_{1,\theta}) + \mathbf{d}\vartheta_0(Z_{2,\theta},\xi_0)\bigg)\,d\theta\nonumber\\
& =  \iota_{\xi_0}\langle \vartheta_2\rangle +\frac{1}{2}\fint \mathbf{d}\vartheta_1^\theta(Z_{1,\theta},\xi_0)\,d\theta+ \fint \bigg(   \frac{1}{2}\mathcal{L}_{Z_{1,\theta}}\mu_1 +\mathcal{L}_{Z_{2,\theta}}\mu_0\bigg)\,d\theta\nonumber\\
& =  \iota_{\xi_0}\langle \vartheta_2\rangle +\frac{1}{2}\fint \mathbf{d}\vartheta_0(\xi_1^\theta,Z_{1,\theta})\,d\theta+ \fint \bigg(   \frac{1}{2}\mathcal{L}_{Z_{1,\theta}}(\mu_1+\mu_1^\theta) +\mathcal{L}_{Z_{2,\theta}}\mu_0\bigg)\,d\theta\nonumber\\
& =  \iota_{\xi_0}\langle \vartheta_2\rangle +\frac{1}{2}\left\langle \mathbf{d}\vartheta_0(\mathcal{L}_{\xi_0}I_0\widetilde{V}_1,I_0\widetilde{V}_1)\right\rangle+ \fint \bigg(   \frac{1}{2}\mathcal{L}_{Z_{1,\theta}}(\mu_1+\mu_1^\theta) +\mathcal{L}_{Z_{2,\theta}}\mu_0\bigg)\,d\theta,
\end{align}
which reproduces Eq.\,\eqref{mu2_formula}.
Lastly, the $O(\epsilon^3)$ terms in Eq.\,\eqref{mu_setup} give
\begin{align}
\mu_3 & = \iota_{\xi_0}\langle\vartheta_3 \rangle  +\fint \mathbf{d}\vartheta_2^\theta (Z_{1,\theta},\xi_0)\,d\theta \nonumber\\
&+ \fint\bigg(\frac{1}{2}\left[\mathcal{L}_{Z_{1,\theta}}\mathbf{d}\vartheta_1^\theta\right] (Z_{1,\theta},\xi_0) + \frac{1}{2} \mathbf{d}\vartheta_1^\theta(Z_{1,\theta},\xi_1)+  \mathbf{d}\vartheta_1^\theta(Z_{2,\theta},\xi_0)\bigg)\,d\theta\nonumber\\
& + \fint\bigg(\frac{2}{3}  \left[\mathcal{L}_{Z_{1,\theta}}\mathbf{d}\vartheta_0\right] (Z_{2,\theta},\xi_0) + \frac{2}{3} \mathbf{d}\vartheta_0 (Z_{2,\theta},\xi_1) + \frac{1}{3} \left[\mathcal{L}_{Z_{1,\theta}}\mathbf{d}\vartheta_0\right](Z_{1,\theta},\xi_1)\nonumber\\
& \quad\quad \quad+ \frac{1}{3}\left[\mathcal{L}_{Z_{2,\theta}}\mathbf{d}\vartheta_0 + \frac{1}{2}\mathcal{L}_{Z_{1,\theta}}^2\mathbf{d}\vartheta_0\right](Z_{1,\theta},\xi_0) + \mathbf{d}\vartheta_0\left(Z_{3,\theta} + \frac{1}{6}[Z_{1,\theta},Z_{2,\theta}],\xi_0\right)\nonumber\\
& \quad \quad \quad + \frac{1}{3}\mathbf{d}\vartheta_0(Z_{1,\theta},\xi_2)\bigg)\,d\theta,
\end{align}
which can again be simplified using 
\begin{gather}
\mathbf{d}\vartheta_2^\theta + \mathcal{L}_{Z_{1,\theta}}\mathbf{d}\vartheta_1^\theta + \mathcal{L}_{Z_{2,\theta}}\mathbf{d}\vartheta_0 + \frac{1}{2}\mathcal{L}_{Z_{1,\theta}}^2\mathbf{d}\vartheta_0 = \mathbf{d}\vartheta_2\\
\iota_{\xi_2}\mathbf{d}\vartheta_0 + \iota_{\xi_1}\mathbf{d}\vartheta_1 + \iota_{\xi_0}\mathbf{d}\vartheta_2 = - \mathbf{d}\mu_2
\end{gather}
together with the identities \eqref{first_order_presymplectic_symmetry}-\eqref{first_order_hamiltonian_mu}, leading to 
\begin{align}
\mu_3  & = \iota_{\xi_0}\langle \vartheta_3\rangle + \frac{2}{3}\fint \mathbf{d}\vartheta_2(Z_{1,\theta},\xi_0)\,d\theta\nonumber\\
& + \frac{1}{3}\fint \mathbf{d}\vartheta_1^\theta(Z_{1,\theta},\xi_1^\theta)\,d\theta -\frac{1}{6}\fint \mathbf{d}\vartheta_1(Z_{1,\theta},\xi_1^\theta)\,d\theta\nonumber\\
& - \frac{1}{3}\fint \mathbf{d}\vartheta_0(Z_{2,\theta},\xi_1^\theta)\,d\theta -\frac{1}{6}\fint \mathbf{d}\vartheta_0(Z_{1,\theta},\mathcal{L}_{Z_{1,\theta}}\xi_1^\theta)\,d\theta\nonumber\\
&+\fint\bigg(\frac{2}{3}\mathcal{L}_{Z_{2,\theta}}\mu_1 + \frac{1}{3}\mathcal{L}_{Z_{2,\theta}}\mu_1^\theta + \frac{1}{6}\mathcal{L}^2_{Z_{1,\theta}}\mu_1^\theta + \frac{1}{3}\mathcal{L}_{Z_{1,\theta}}\mu_2\bigg)\,d\theta\nonumber\\
& + \fint \mathcal{L}_{Z_{3,\theta} + \frac{1}{6}[Z_{1,\theta},Z_{2,\theta}]}\mu_0\,d\theta.\label{mu3_Z}
\end{align}
Using the identity
\begin{align}
\fint \mathbf{d}\vartheta_0([I_0\widetilde{V}_1,I_0\widetilde{V}_1^\theta],\mathcal{L}_{\xi_0}I_0\widetilde{V}_1^\theta)\,d\theta=&-\frac{1}{2}\mathcal{L}_{I_0\widetilde{V}_1}\langle \mathbf{d}\vartheta_0(\mathcal{L}_{\xi_0}I_0\widetilde{V}_1,I_0\widetilde{V}_1)\rangle \nonumber\\
&- \frac{1}{2}\mathbf{d}\widetilde{\vartheta}_1\cdot\left\langle \mathcal{L}_{\xi_0}I_0\widetilde{V}_1\wedge I_0\widetilde{V}_1 \right\rangle,
\end{align}
together with Eqs.\,\eqref{xi2_formula} and \eqref{Z2_explicit}, the first five integrals in Eq.\,\eqref{mu3_Z} may be evaluated explicitly, resulting in the desired expression \eqref{mu3_formula}.

\end{proof}

\section{Example 1: charged particle in a magnetic field\label{charged_particle_example}}
As an example and verification test for the formulas provided by Theorem \ref{main_result}, we will now use Theorem \ref{main_result} to recover the first two terms of the well-known adiabatic invariant series for a charged particle in a magnetic field. The nearly-periodic Hamiltonian system that describes such charged particles is the ODE on $Q\times\mathbb{R}^3$ given by
\begin{align}
\dot{\bm{v}}& = \frac{1}{\epsilon}\bm{v}\times\bm{B}(\bm{x})\nonumber\\
\dot{\bm{x}}& = \bm{v},
\end{align}
where $Q\subset \mathbb{R}^3$ is an open subset representing the spatial domain, and $\bm{B} =\nabla\times\bm{A}$ is a magnetic field on $Q$. If $\bm{b} = \bm{B}/|\bm{B}|$ denotes the unit vector along the magnetic field then the limiting roto-rate vector is given by $\xi_0 = \bm{v}\times\bm{b}\cdot \partial_{\bm{v}}$, the frequency function $\omega_0 = |\bm{B}|$, and the Hamiltonian structure is specified by the one-form $\vartheta_\epsilon = \bm{A}\cdot d\bm{x} +\epsilon \bm{v}\cdot d\bm{x}$ and the Hamiltonian $H_\epsilon = \epsilon \frac{1}{2}|\bm{v}|^2$. The exponential of the limiting roto-rate vector is given by
\begin{align}
\exp(\theta\xi_0)(\bm{x},\bm{v}) = (\bm{x}, \bm{v}\cdot\bm{b}\,\bm{b} + \sin\theta\,\bm{v}\times\bm{b} + \cos\theta\,\bm{b}\times(\bm{v}\times\bm{b})),
\end{align}
where $\bm{b}$ should be evaluated at $\bm{x}$. 

Consider first $\mu_0$, which according to Theorem \ref{main_result} is given by Eq.\,\eqref{mu0_formula}. Because the flow of $\xi_0$ leaves $\bm{x}$ unchanged the average $\langle\vartheta_0\rangle = \vartheta_0 = \bm{A}\cdot d\bm{x}$. Therefore $\mu_0 = \iota_{\xi_0}\vartheta_0 = (\bm{A}\cdot d\bm{x})(\bm{v}\times\bm{b}\cdot\partial_{\bm{v}}) =0$. This says that the adiabatic invariant series for this system is degenerate to leading-order.

Next consider $\mu_1$, which according to Theorem \ref{main_result} is given by Eq.\,\eqref{mu1_formula}. Since $\mu_0$ vanishes, the general formula simplifies to $\mu_1 = \iota_{\xi_0}\langle \vartheta_1\rangle$, where $\vartheta_1 = \bm{v}\cdot d\bm{x}$. The average of this one-form is given by
\begin{align}
\langle \vartheta_1\rangle = \fint \exp(\theta\,\xi_0)^*\vartheta_1\,d\theta = \fint \bigg(\bm{v}_\theta\cdot d\bm{x}\bigg)\,d\theta = (\bm{v}\cdot\bm{b})\,\bm{b}\cdot d\bm{x},
\end{align}
where we have introduced the shorthand $\bm{v}_\theta =\bm{v}\cdot\bm{b}\,\bm{b} + \sin\theta\,\bm{v}\times\bm{b} + \cos\theta\,\bm{b}\times(\bm{v}\times\bm{b}) $. Therefore the first-order term in the adiabatic invariant series is $\mu_1 = ((\bm{v}\cdot\bm{b})\bm{b}\cdot d\bm{x})(\bm{v}\times\bm{b}\cdot\partial_{\bm{v}}) = 0$. A double degeneracy!

The calculation starts to get interesting with $\mu_2$. Due to the double degeneracy, and the fact that $\vartheta_2 = 0$, the general formula \eqref{mu2_formula} simplifies to $\mu_2 = \frac{1}{2}\left\langle \mathbf{d}\vartheta_0(\mathcal{L}_{\xi_0}I_0\widetilde{V}_1,I_0\widetilde{V}_1)\right\rangle$. For the sake of evaluating this expression it is useful to record the following formulas for $V_1^\theta$, $I_0\widetilde{V}_1^\theta$, and $\mathcal{L}_{\xi_0}I_0\widetilde{V}_1^\theta$,
\begin{align}
{V}_1^\theta & =( \bm{v}\cdot\bm{b}\bm{b})\cdot\partial_{\bm{x}} + \frac{1}{2}\left\{([\bm{v}\times\bm{b}]\cdot\nabla\bm{b})\times\bm{v} + 2(\bm{v}\cdot\bm{b})(\bm{b}\times\bm{\kappa})\times\bm{v}-(\bm{b}\times[\bm{v}_\perp\cdot\nabla\bm{b}])\times\bm{v}\right\}\cdot\partial_{\bm{v}}\nonumber\\
& + \cos\theta\left(\bm{v}_\perp\cdot\partial_{\bm{x}} + \left\{(\bm{b}\times [\bm{v}_\perp\cdot\nabla\bm{b}])\times\bm{v} - (\bm{v}\cdot\bm{b})(\bm{b}\times\bm{\kappa})\times\bm{v}\right\}\cdot\partial_{\bm{v}}\right)\nonumber\\
& + \sin\theta\left(\bm{v}\times\bm{b}\cdot\partial_{\bm{x}} + \left\{(\bm{v}\cdot\bm{b})\bm{\kappa}\times\bm{v} + (\bm{b}\times[[\bm{v}\times\bm{b}]\cdot\nabla\bm{b}])\times\bm{v}\right\}\cdot\partial_{\bm{v}}\right)\nonumber\\
&-\frac{1}{2}\cos2\theta\,\{([\bm{v}\times\bm{b}]\cdot\nabla\bm{b})\times\bm{v} + (\bm{b}\times[\bm{v}_\perp\cdot\nabla\bm{b}])\times\bm{v}\}\cdot\partial_{\bm{v}}\nonumber\\
&+\frac{1}{2}\sin2\theta\,\left\{(\bm{v}_\perp\cdot\nabla\bm{b})\times\bm{v} - (\bm{b}\times[[\bm{v}\times\bm{b}]\cdot\nabla\bm{b}])\times\bm{v}\right\}\cdot\partial_{\bm{v}}\label{gc_V1theta}\\
I_0\widetilde{V}_1^\theta & = -|\bm{B}|^{-1}[\bm{v}_\perp\cos\theta + \bm{v}\times\bm{b}\sin\theta]\cdot\nabla\text{ln}|\bm{B}| (\bm{v}\times\bm{b})\cdot\partial_{\bm{v}}\nonumber\\
& + |\bm{B}|^{-1}\sin\theta\left(\bm{v}_\perp\cdot\partial_{\bm{x}} + \left\{(\bm{b}\times [\bm{v}_\perp\cdot\nabla\bm{b}])\times\bm{v} - (\bm{v}\cdot\bm{b})(\bm{b}\times\bm{\kappa})\times\bm{v}\right\}\cdot\partial_{\bm{v}}\right)\nonumber\\
& -|\bm{B}|^{-1} \cos\theta\left(\bm{v}\times\bm{b}\cdot\partial_{\bm{x}} + \left\{(\bm{v}\cdot\bm{b})\bm{\kappa}\times\bm{v} + (\bm{b}\times[[\bm{v}\times\bm{b}]\cdot\nabla\bm{b}])\times\bm{v}\right\}\cdot\partial_{\bm{v}}\right)\nonumber\\
&-\frac{1}{4}|\bm{B}|^{-1}\sin2\theta\,\{([\bm{v}\times\bm{b}]\cdot\nabla\bm{b})\times\bm{v} + (\bm{b}\times[\bm{v}_\perp\cdot\nabla\bm{b}])\times\bm{v}\}\cdot\partial_{\bm{v}}\nonumber\\
&-\frac{1}{4}|\bm{B}|^{-1}\cos2\theta\,\left\{(\bm{v}_\perp\cdot\nabla\bm{b})\times\bm{v} - (\bm{b}\times[[\bm{v}\times\bm{b}]\cdot\nabla\bm{b}])\times\bm{v}\right\}\cdot\partial_{\bm{v}}\label{gc_IV1theta}\\
\mathcal{L}_{\xi_0}I_0\widetilde{V}_1^\theta & = -|\bm{B}|^{-1}[\bm{v}\times\bm{b}\cos\theta-\bm{v}_\perp\sin\theta ]\cdot\nabla\text{ln}|\bm{B}| (\bm{v}\times\bm{b})\cdot\partial_{\bm{v}}\nonumber\\
& + |\bm{B}|^{-1}\cos\theta\left(\bm{v}_\perp\cdot\partial_{\bm{x}} + \left\{(\bm{b}\times [\bm{v}_\perp\cdot\nabla\bm{b}])\times\bm{v} - (\bm{v}\cdot\bm{b})(\bm{b}\times\bm{\kappa})\times\bm{v}\right\}\cdot\partial_{\bm{v}}\right)\nonumber\\
& +|\bm{B}|^{-1} \sin\theta\left(\bm{v}\times\bm{b}\cdot\partial_{\bm{x}} + \left\{(\bm{v}\cdot\bm{b})\bm{\kappa}\times\bm{v} + (\bm{b}\times[[\bm{v}\times\bm{b}]\cdot\nabla\bm{b}])\times\bm{v}\right\}\cdot\partial_{\bm{v}}\right)\nonumber\\
&-\frac{1}{2}|\bm{B}|^{-1}\cos2\theta\,\{([\bm{v}\times\bm{b}]\cdot\nabla\bm{b})\times\bm{v} + (\bm{b}\times[\bm{v}_\perp\cdot\nabla\bm{b}])\times\bm{v}\}\cdot\partial_{\bm{v}}\nonumber\\
&+\frac{1}{2}|\bm{B}|^{-1}\sin2\theta\,\left\{(\bm{v}_\perp\cdot\nabla\bm{b})\times\bm{v} - (\bm{b}\times[[\bm{v}\times\bm{b}]\cdot\nabla\bm{b}])\times\bm{v}\right\}\cdot\partial_{\bm{v}}\label{gc_LIV1theta}.
\end{align}
where $\bm{\kappa} = \bm{b}\cdot\nabla\bm{b}$ is the field line curvature and $\bm{v}_\perp = \bm{b}\times(\bm{v}\times\bm{b})$ is the projection of the velocity into the plane perpendicular to the magnetic field.
If $U = U_{\bm{x}}\cdot\partial_{\bm{x}} + U_{\bm{v}}\cdot\partial_{\bm{v}}$ and $W = W_{\bm{x}}\cdot\partial_{\bm{x}} + W_{\bm{v}}\cdot\partial_{\bm{v}}$ are any two vector fields on $Q\times\mathbb{R}^3$ then $\mathbf{d}\vartheta_0(U,V) = \bm{B}\cdot U_{\bm{x}}\times W_{\bm{x}}$. In particular, when $U$ is given by Eq.\,\eqref{gc_LIV1theta} and $W$ is given by Eq.\,\eqref{gc_IV1theta} the expression becomes  $\mathbf{d}\vartheta_0(\mathcal{L}_{\xi_0}I_0\widetilde{V}_1^\theta,I_0\widetilde{V}_1^\theta) = |\bm{B}|^{-1}|\bm{v}\times\bm{b}|^2 $ for each $\theta\in S^1$. It follows that $\mu_2 = \frac{1}{2}|\bm{B}|^{-1}|\bm{v}\times\bm{b}|^2$, which is the familiar expression for the leading term in the magnetic moment series. Note that because $\mu_2$ is the first nontrivial term in the adiabatic invariant series for this system standard convention is to refer to this quantity as $\mu_0$ rather than $\mu_2$. We have not adopted this convention in this article because not all nearly-periodic Hamiltonian systems exhibit the double degeneracy $\mu_0 = \mu_1 = 0$.

Finally consider $\mu_3$, which should give the first correction to the magnetic moment. Because $\vartheta_2 = \vartheta_3 = 0$, $V_2 = 0$, $\mu_0 = 0$, and $\mu_1 = 0$ the general formula \ref{mu3_formula} reduces to
\begin{align}
\mu_3& = - \frac{1}{3}\langle \mathbf{d}\vartheta_1(\mathcal{L}_{\xi_0}I_0\widetilde{V}_1,I_0\widetilde{V}_1)\rangle\nonumber\\
& + \frac{1}{3}\left\langle \iota_{\mathcal{L}_{\xi_0}I_0\widetilde{V}_1}\mathbf{d}\vartheta_1\right\rangle(I_0\widetilde{V}_1) - \frac{1}{6}\langle \mathbf{d}\vartheta_0([\mathcal{L}_{\xi_0}I_0\widetilde{V}_1,I_0\widetilde{V}_1], I_0\widetilde{V}_1)\rangle \nonumber\\
&+ \frac{1}{6}\mathbf{d}\vartheta_0(\langle [\mathcal{L}_{\xi_0}I_0\widetilde{V}_1,I_0\widetilde{V}_1]\rangle,I_0\widetilde{V}_1) \nonumber\\
& + \frac{1}{3}\langle \mathbf{d}\vartheta_0(\mathcal{L}_{\xi_0}I_0\widetilde{V}_1, I_0[I_0\widetilde{V}_1,\langle V_1\rangle])\rangle+\frac{1}{6}\langle \mathbf{d}\vartheta_0(\mathcal{L}_{\xi_0}I_0\widetilde{V}_1,I_0[I_0\widetilde{V}_1,\widetilde{V}_1]^{\text{osc}})\rangle\nonumber\\
& + \frac{1}{6}\langle  \langle\mathbf{d}\vartheta_1 \rangle(\mathcal{L}_{\xi_0}I_0\widetilde{V}_1,I_0\widetilde{V}_1)\rangle - \frac{1}{3}\mathcal{L}_{I_0\widetilde{V}_1}\langle \mathbf{d}\vartheta_0(\mathcal{L}_{\xi_0}I_0\widetilde{V}_1,I_0\widetilde{V}_1) \rangle. \label{mu3_gc}
\end{align}
In order to eliminate the possibility of human errors in evaluating each of the terms in \eqref{mu3_gc} we used the vector calculus simplification tool \emph{VEST} to perform the calculation. \emph{VEST} was originally developed in \cite{Squire_2013} for the purpose of implementing the automatic calculation of the guiding center calculation in \cite{Burby_gc_2013}, and is therefore admirably suited to the present calculation. The final result is
\begin{align}
\mu_3 &= \mu_0\frac{(\bm{b}\times\bm{v})\cdot \nabla |\bm{B}|}{|\bm{B}|^2} + \frac{1}{4}\frac{(\bm{v}\cdot\bm{b})\,\bm{v}\cdot\nabla\bm{b}\cdot(\bm{v}\times\bm{b})}{|\bm{B}|^2} \nonumber\\
&-\frac{3}{4}\frac{(\bm{v}\cdot\bm{b})\, (\bm{v}\times\bm{b})\cdot\nabla\bm{b}\cdot\bm{v} }{|\bm{B}|^2} - \frac{5}{4}\frac{(\bm{v}\cdot\bm{v})^2\,\bm{\kappa}\cdot(\bm{v}\times\bm{b})}{|\bm{B}|^2},
\end{align}
which agrees with the formula from \cite{Weyssow_1986}.


\section{Example 2: an adiabatic invariant for nearly-periodic magnetic fields\label{field_line_invariant_sec}}
Kolmogorov-Arnold-Moser (KAM) theory reveals much about the structure of toroidal magnetic fields used for the purpose of magnetic confinement fusion. Perhaps most significantly it provides the following stability result. If the true magnetic field within a device is close to a fiducial field with nested toroidal flux surfaces, and the magnetic shear of the fiducial field is bounded away from zero, then the true field will have nearly the same measure of flux surfaces as the fiducial field. In the narrow gaps between the surviving flux surfaces, deterministic chaos reigns.

On the other hand, KAM theory says very little when the fiducial, unperturbed field has vanishing shear, particularly when the rotational transform is constant and rational. For example \cite{10.2307/j.ctt1bd6kg5} requires perturbations to be small relative to shear. Finite shear ensures that many of the unperturbed flux surfaces possess strongly non-resonant (i.e. strongly irrational) rotational transform. Such non-resonant tori survive perturbations with relative ease, and provide the true, perturbed magnetic field with its source of KAM tori. When all of the fiducial field lines are closed, however, this well of non-resonant unperturbed flux surfaces runs dry. Even though the unperturbed field contains many (non-unique) flux surfaces, each of these is strongly resonant. It would therefore seem that closed-line fields should easily be blown apart by most perturbations.

In order to critically examine the validity of this last statement, suppose that $\bm{B}_\epsilon = \nabla\times\bm{A}_\epsilon$ is a non-vanishing magnetic field for each $\epsilon\in\mathbb{R}$. The field $\bm{B}_\epsilon$ has \emph{nearly-closed field lines} if $\bm{A}_\epsilon$ depends smoothly on $\epsilon$ and each field line of $\bm{B}_0$ is closed. A remarkable property of such a magnetic field is that the associated field-line dynamical system $\dot{\bm{x}} = \epsilon^{-1}\bm{B}_\epsilon(\bm{x})$ comprises a nearly-periodic Hamiltonian system on $Z = Q$, the field-line container. The frequency function is given by
\begin{align}
\frac{1}{\omega_0(\bm{x})} = \frac{1}{2\pi}\oint_{\ell_0(\bm{x})}\frac{d\ell}{|\bm{B}|},
\end{align}
where $\ell_0(\bm{x})$ is the unique $\bm{B}_0$-line that contains the point $\bm{x}\in Q$; the limiting roto-rate vector is $\xi_0 = \bm{B}_0/\omega_0$; the presymplectic form is $-\mathbf{d}(\bm{A}_\epsilon\cdot d\bm{x}) = -\iota_{\bm{B}_\epsilon}d^3\bm{x}$; and the Hamiltonian is $H_\epsilon =0$. Therefore the general theory outlined in \cite{Kruskal_1962}, as well as the rest of this Article, guarantees the existence of a field-line adiabatic invariant $\mu_\epsilon$ for $\bm{B}_\epsilon$ with $\epsilon$ small but finite. Since such an adiabatic invariant defines approximate flux surfaces (surfaces that field lines traverse many times before possibly wandering away), magnetic fields with nearly-closed lines of force enjoy much more stability that KAM theory, and in particular that theory's assumption of non-vanshing shear suggests. Note in particular that existence of the field-line adiabatic invariant $\mu_\epsilon$ does \emph{not} require the perturbation $\bm{B}_\epsilon - \bm{B}_0$ to be non-resonant. 

The robustness of magnetic fields with nearly-closed field lines is consistent with previous experimental and theoretical analyses of magnetic fields and fluid flows with regions of low or sign-reversing shear. See for example \cite{Firpo_2011} for a fusion-oriented study and \cite{Negrete_1992} for an investigation of analogous ideas in the context of sheared fluid flow. In any low-shear region a rational number $q/p$ may be found that uniformly approximates the unperturbed field's rotational transform $\iota(\psi)$. By writing $\iota(\psi) = q/p + (\iota(\psi) - q/p)$ the unperturbed magnetic field may then be expressed as a field with closed lines plus a correction that is proportional to the shear. Because the shear is small by hypothesis it is therefore natural to lump the correction term $\delta\iota(\psi )  =\iota(\psi) - q/p$ together with any magnetic perturbations that may be present. In this manner the magnetic field within a region of small shear may be expressed as a magnetic field with nearly-closed field lines. The field line dynamics within a low-shear region therefore possess an adiabatic invariant $\mu_\epsilon$. Moreover approximate level sets of $\mu_\epsilon$ may be used to quickly predict the transport effects, deleterious or not, of the perturbation. This approach to understanding the impacts of perturbations on low-shear magnetic fields appears to have gone largely unnoticed; the preferred approach has been the more-cumbersome and obtuse action-angle formalism.


In order to gain insight into the significance of the adiabatic flux surfaces defined by $\mu_\epsilon$, consider the formulas for $\mu_\epsilon$ provided by Theorem \ref{main_result}. According to Eq.\,\eqref{mu0_formula} the coefficient $\mu_0=\text{const.}$ since $\mathbf{d}\mu_0 = -\iota_{\xi_0}\mathbf{d}\vartheta_0 = \omega_0^{-1}\iota_{\bm{B}_0}\iota_{\bm{B}_0}d^3\bm{x} = 0$. Therefore the first possibly non-trivial coefficient is $\mu_1 = \iota_{\xi_0}\langle \vartheta_1\rangle$. Apparently the value of $\mu_1$ at $\bm{x}\in Q$ may be written as the line integral $\mu_1(\bm{x}) = (2\pi)^{-1}\oint_{\ell_0(\bm{x})}\bm{A}_1\cdot d\bm{x} $, with $\ell_0(\bm{x})$ defined as above and oriented so that $\bm{B}(\bm{x})$ is a positive basis for $T_{\bm{x}}\ell_0(\bm{x})$. This quantity can be understood as a magnetic flux as follows. Suppose that $Q$ is path connected, fix $\bm{x}_0\in Q$, and define the constant $\mu^*_\epsilon = (2\pi)^{-1}\oint_{\ell_0(\bm{x}_0)}\bm{A}_\epsilon\cdot d\bm{x}$. The quantity $\overline{\mu}_\epsilon = \mu_\epsilon - \mu_\epsilon^*$ is an adiabatic invariant since it differs from $\mu_\epsilon$ by a constant. The first non-zero coefficient of $\overline{\mu}_\epsilon$ is $\overline{\mu}_1(\bm{x}) = \mu_1(\bm{x}) - \mu_1^*$, which is, up to a constant, the same as $\mu_1$. Let $\bm{x}(\lambda)$ be a curve in $Q$ with, $\partial_\lambda\bm{x}(\lambda)\neq 0$, $\bm{x}(0) = \bm{x}_0$, and $\bm{x}(1) = \bm{x}$. Set $R_0(\bm{x}) = \cup_{\lambda\in[0,1]}\ell_0(\bm{x}(\lambda))$. Note that $R_0(\bm{x})$ is a flux ribbon for $\bm{B}_0$ because it is a union of $\bm{B}_0$-lines. Now apply Stoke's theorem according to
\begin{align}
\overline{\mu}_1(\bm{x})& = \frac{1}{2\pi}\oint_{\ell_0(\bm{x})}\bm{A}_1\cdot d\bm{x} - \frac{1}{2\pi}\oint_{\ell_0(\bm{x}_0)}\bm{A}_1\cdot d\bm{x}\nonumber\\
 & = \frac{1}{2\pi}\int_{R_0(\bm{x})}\bm{B}_1\cdot d\bm{S},
\end{align}
where $R_0(\bm{x})$ is oriented so that $(\partial_\lambda\bm{x}(\lambda),\bm{B}(\bm{x}(\lambda))$ is a positive basis for $T_{\bm{x}(\lambda)}R_0(\bm{x})$. This shows that, up to an unimportant additive constant, $\mu_1(\bm{x})$ is equal to the (normalized) flux of $\bm{B}_1$ through any flux ribbon $R_0(\bm{x})$ whose boundary is $\partial R_0(\bm{x}) = \ell_0(\bm{x})\cup \ell_0(\bm{x}_0)$. If $\bm{B}_0$ contains a single magnetic axis $L$ then it is permissible to set $\ell_0(\bm{x}_0) = L$. In this special case $2\pi\overline{\mu}_1$ is a perturbed poloidal flux.

The approximate flux surfaces defined by the level sets of $\mu_1$ (or equivalent $\overline{\mu}_1$) determine how well field lines are confined within $Q$. The most favorable case for confinement occurs when the $\mu_1$-surfaces are nested tori contained in $Q$, but more exotic foliations may occur depending on the form of the perturbation $\bm{B}_1$. For example, let $\bm{B}_\epsilon = \bm{B}_0 + \epsilon \,\bm{B}_1$ with $\bm{B}_0 = B_0 (R_0/R)\bm{e}_{\phi}$ a tokamak vacuum field and $\bm{B}_1 = \alpha \nabla \psi\times\nabla\phi$. Here $(R,\phi,Z)$ are standard cylindrical coordinates, $\psi$ is an arbitrary function, and $\alpha$ is a constant. The function $\mu_1$ is then
\begin{align}
\mu_1(R,\phi,Z) = \alpha\frac{1}{2\pi}\int_0^{2\pi}\psi(R,\overline{\phi},Z)\,d\overline{\phi} = \alpha\langle \psi\rangle(R,Z),
\end{align}
where the angle brackets denote an azimuthal average. In this example adiabatic flux surfaces are surfaces of revolution with poloidal cross sections given by level sets of $\langle \psi\rangle(R,Z)$. If $\psi(R,Z) = (R - R_0)^2 + Z^2$ the poloidal cross sections are nested circles centered at $(R,Z) = (R_0,0)$, indicating confinement. However, if $\psi(R,Z) = (R-R_0)Z$ the cross sections are hyperbolas, indicating no such confinement.

\section{Conclusion}
In this Article we have succeeded in deriving and verifying general coordinate-independent expressions for the adiabatic invariant associated with a nearly-periodic Hamiltonian system. These formulas are summarized in Theorem \ref{main_result}. As a byproduct of our derivation we have also derived coordinate-independent expressions for the roto-rate vector associated with a possibly-non-Hamiltonian nearly-periodic system. These formulas are summarized in Theorem \ref{roto_formula_thm}. Using these formulas, adiabatic invariants may be computed more efficiently and directly than prior procedure-based methods.

A goal of future work will be to apply our results to infinite dimensional systems, especially systems with slow manifolds such as ideal magnetohydrodynamics, \cite{Burby_two_fluid_2017}, kinetic MHD, \cite{Burby_Sengupta_2017_pop}, and Lorentz loop dynamics, \cite{Burby_jmp_loops_2020}. See \cite{Burby_review_2020} for an in-depth discussion of the role of slow manifolds in plasma physics.  Just as \cite{Cotter_2004} shows that the long-time persistence of quasigeostrophic balance in non-dissipative geophysical fluid flows may be explained by finding an appropriate adiabatic invariant, adiabatic invariants in these plasma-dynamical systems may explain subtle notions such as the persistence timescale for gyrotropy in strongly-magnetized plasmas. The key concept underlying the results of \cite{Cotter_2004} is the identification of quasigeostrophic dynamics with motion on a slow manifold. We remark that the relationship between slow manifolds and quasigeostrophic balance was established previously in \cite{Lorenz_1986}, \cite{Lorenz_1987}, and \cite{Lorenz_1992}.

\section{Acknowledgements}
Research presented in this article was supported by the Los Alamos National Laboratory LDRD program under project number 20180756PRD4. Support for JS was provided by Rutherford Discovery Fellowship RDF-U001804 and Marsden Fund grant UOO1727, which are managed through the Royal Society Te Ap\=arangi. 

\appendix
\section{How-to guide for the new formulas\label{app:how_to}}
In this appendix we provide explicit details on how to practically compute the various terms in Eqs.\,\eqref{mu0_formula} - \eqref{mu3_formula}. 

As written, these formulas involve vector fields like $V$, one-forms like $\vartheta$, and two-forms like $\mathbf{d}\vartheta$. For those unfamiliar with exterior calculus, the index-notation equivalents of these objects are summarized as follows:
\begin{align}
V&\leftrightarrow V^i\quad\text{(vector fields)}\\
\vartheta &\leftrightarrow \vartheta_i\quad\text{(one-forms)}\\
\mathbf{d}\vartheta & \leftrightarrow (\mathbf{d}\vartheta)_{ij}\quad \text{(two-forms)}\\
\end{align}
The components of a two-form like $\mathbf{d}\vartheta$ are related to the components of $\vartheta$ according to
\begin{align}
(\mathbf{d}\vartheta)_{ij} &= \partial_i\vartheta_j - \partial_j\vartheta_i, \label{dext}
\end{align}
and are therefore anti-symmetric. Actually all two-forms are antisymmetric. A one-form $\vartheta$ can be contracted with a single vector field $V$ in order to produce a scalar field $\vartheta(V)$. A two-form $\mathbf{d}\vartheta$ can be contracted with two vector fields $V_1,V_2$ in order to produce a scalar field, $\mathbf{d}\vartheta(V_1,V_2)$. In index notation these contraction operations are summarized as follows,
\begin{align}
\vartheta(V)& = \vartheta_i\,V^i\\
\mathbf{d}\vartheta(V_1,V_2) & = (\mathbf{d}\vartheta)_{ij}\,V_1^i\,V_2^j.
\end{align}
A two-form $\mathbf{d}\vartheta$ may also be contracted on the left with a single vector field $V$ to obtain a one-form $\iota_V\mathbf{d}\vartheta$ given by
\begin{align}
(\iota_V\mathbf{d}\vartheta)_i = V^j(\mathbf{d}\vartheta)_{ji}.
\end{align}

As for calculus, there are two important operations that must be handled. The commutator of two vector fields $[V_1,V_2]$ is given by
\begin{align}
[V_1,V_2]^i & = V_1^j\partial_j V_2^i - V_2^j\partial_j V_1^i.
\end{align}
The Lie derivative of a scalar $\mu$ along a vector field $V$, $\mathcal{L}_V\mu$, is given by
\begin{align}
\mathcal{L}_V\mu = V^i\partial_i\mu.
\end{align}

One non-trivial operation that must be performed when evaluating the formulas \eqref{mu0_formula} - \eqref{mu3_formula} is the $U(1)$-average, denoted $\langle \cdot\rangle$. In order to carry out this operation analytically it is necessary to have an explicit expression for the phase-space mappings $\zeta_\theta:z\mapsto \exp(\theta\,\xi_0)(z)\equiv \zeta_\theta(z)$. Practically-speaking, the value of $\exp(\theta\,\xi_0)(z)$ is $z(\theta)$, where $z(\theta)$ is the unique solution of the ODE $\partial_\theta z(\theta) = \xi_0(z(\theta))$ with $z(0) = z$. Therefore knowledge of the mapping $\zeta_\theta$ is tantamount to knowledge of the general solution of the ODE defined by $\xi_0$. Notice that $\xi_0$-trajectories are just reparameterizations of the leading-order dynamical trajectories $\dot{z} = V_0(z)$. It may be helpful to find a coordinate system where $\xi_0$ is simple in order to find an explicit expression $\zeta_\theta$. Two examples of $ \zeta_\theta(z)$ were given in the text. Once $\zeta_\theta(z)$ is known, the $U(1)$-average can be applied to any tensor, in particular vector fields $V$ and one-forms $\vartheta$. In components, the relevant formulas are
\begin{align}
\langle V\rangle^i(z) & = \fint V^j(\zeta_\theta(z))\,\partial_j\zeta^i_{-\theta}(\zeta_\theta(z)) \,d\theta \\
\langle \vartheta\rangle_i(z) & =\fint \vartheta_j(\zeta_\theta(z))\,\partial_i\zeta_\theta^j(z)\,d\theta.\label{oneformaverage}
\end{align}
Because the $U(1)$-average commutes with the exterior derivative, $\mathbf{d}$, the average of a two-form like $\mathbf{d}\vartheta$ may be computed by first finding $\langle \vartheta\rangle$ using Eq.\,\eqref{oneformaverage} and then computing $\mathbf{d}\langle \vartheta\rangle = \langle \mathbf{d}\vartheta \rangle$ using Eq.\,\eqref{dext}.

The least non-trivial operation encountered in Eqs.\,\eqref{mu0_formula}-\eqref{mu3_formula} is the operator $I_0 = (\mathcal{L}_{V_0})^{-1}$. In a rough sense this operator ``integrates along unperturbed orbits." The easiest way to compute $I_0\widetilde{V}$ for a fluctuating vector field $\widetilde{V}$ is to use the following Fourier-series-based trick. First compute $\widetilde{V}^\theta = \zeta_\theta^*\widetilde{V}$, which in components is given by
\begin{align}
\widetilde{V}^{\theta i}(z) = \widetilde{V}^j(\zeta_\theta(z))\partial_j\zeta^i_{-\theta}(\zeta_\theta(z)).
\end{align}
Because $\zeta_\theta(z)$ is $2\pi$-periodic in $\theta$ the component $\widetilde{V}^{\theta i}(z)$ must have a Fourier series expansion
\begin{align}
\widetilde{V}^{\theta i}(z) = \sum_{n\in\mathbb{Z}} \widetilde{V}^i_n(z)\,e^{in\theta},
\end{align}
where the Fourier coefficients $\widetilde{V}^i_n(z)$ are complex-valued functions of $z$. We can use these Fourier coefficients to help solve the problem $\mathcal{L}_{V_0}\widetilde{U} = \widetilde{V}$ for $\widetilde{U}$ given $\widetilde{V}$. Note that finding $\widetilde{U}$ is equivalent to finding $I_0\widetilde{V}$. To see how, apply the pullback $\zeta_\theta^*$ to the equation $\mathcal{L}_{V_0}\widetilde{U} = \widetilde{V}$ to obtain $\mathcal{L}_{V_0}\widetilde{U}^\theta = \widetilde{V}^\theta$. The LHS of this equation simplifies considerably by noting $\mathcal{L}_{V_0}\widetilde{U}^\theta = -\mathcal{L}_{\widetilde{U}^\theta}V_0 =-\mathcal{L}_{\widetilde{U}^\theta}(\omega_0 \xi_0) = -\xi_0\,\mathcal{L}_{\widetilde{U}^\theta}\omega_0 +\omega_0\mathcal{L}_{\xi_0}\widetilde{U}^\theta = -\xi_0\,\mathcal{L}_{\widetilde{U}^\theta}\omega_0 +\omega_0\partial_\theta\widetilde{U}^\theta$. In components this identity may be written 
\begin{align}
(\mathcal{L}_{V_0}\widetilde{U}^\theta)^i & = -\widetilde{U}^{\theta j}\partial_j\omega_0\,\xi^i + \omega_0\,\partial_\theta\widetilde{U}^{\theta i}. 
\end{align}
Therefore the Fourier components of the equation $\mathcal{L}_{V_0}\widetilde{U}^\theta = \widetilde{V}^\theta$ read
\begin{align}
-\widetilde{U}^{ j}_n\partial_j\omega_0\,\xi^i + i\,n\,\omega_0\widetilde{U}^{ i}_n = \widetilde{V}^i_n.
\end{align}
This infinite sequence of algebraic equations may be solved by hand, leading to the following formula for $\widetilde{U}^i_n$,
\begin{align}
\widetilde{U}^i_n = \xi^i \, \frac{\widetilde{V}^j_n \partial_j\omega_0}{(in\omega_0)^2} + \frac{\widetilde{V}_n^i}{in\omega_0}.
\end{align}
Since $\widetilde{U}^\theta$ is equal to $\widetilde{U} = I_0\widetilde{V}$ when $\theta = 0$, it follows that $I_0$ is determined by the formula
\begin{align}
(I_0\widetilde{V})^i = \sum_{n\in\mathbb{Z}} \xi^i \, \frac{\widetilde{V}^j_n \partial_j\omega_0}{(in\omega_0)^2} +\sum_{n\in\mathbb{Z}} \frac{\widetilde{V}_n^i}{in\omega_0},
\end{align}
which can always be used to compute $I_0$.

\section{Derivation of the formula for $\xi_3$\label{appA}}
In the proof of Theorem \ref{roto_formula_thm} we derived the general formula \eqref{xi3_general_formula} for $\widetilde{\xi}_3$, but did not show how that formula can be manipulated in order to produce Eq.\,\eqref{xi3_formula} for $\xi_3$. In this Appendix we will complete the demonstration. The required manipulations are based on recursive applications of the Leibniz rule for the bracket of vector fields, i.e. the Jacobi identity. In spirit, such identities are similar to the well-known recursive Leibniz identity
\[
e^{x}\sin x = - \frac{d}{dx}(e^x \cos x) + e^x \cos x = -\frac{d}{dx} (e^x\cos x - e^x\sin x) - e^x\sin x.
\]

Starting from the general formula for $\widetilde{\xi}_3$ from Theorem \ref{roto_formula_thm},

\begin{align}
\widetilde{\xi}_3 
& = \Lxio I_0\widetilde{V}_3 + I_0[\Lxio I_0\widetilde{V}_1,V_2]^{\text{osc}}  + I_0[\Lxio I_0\widetilde{V}_2,V_1]^{\text{osc}} \nonumber\\
&+ I_0\bigg[ \Lxio I_0[ I_0\widetilde{V}_1, \langle V_1\rangle] + \frac{1}{2}\Lxio I_0[I_0\widetilde{V}_1,\widetilde{V}_1]^{\text{osc}} + \frac{1}{2}[\Lxio I_0\widetilde{V}_1, I_0\widetilde{V}_1],V_1\bigg]^{\text{osc}} \nonumber\\
& = \Lxio I_0\widetilde{V}_3 + \Lxio I_0[ I_0\widetilde{V}_1,\langle V_2\rangle]^{\text{osc}}  + \Lxio I_0[ I_0\widetilde{V}_2,\langle V_1\rangle]^{\text{osc}} \nonumber\\
&+  \frac{1}{2}\Lxio I_0[I_0\widetilde{V}_2,\widetilde{V}_1]^{\text{osc}} + \frac{1}{2} \Lxio  I_0[I_0\widetilde{V}_1,\widetilde{V}_2]^{\text{osc}}\nonumber\\
&+  \frac{1}{2} [\Lxio I_0 \widetilde{V}_1,I_0\widetilde{V}_2]^{\text{osc}}+  \frac{1}{2} [\Lxio I_0 \widetilde{V}_2,I_0\widetilde{V}_1]^{\text{osc}}\nonumber\\
&+ I_0\bigg[ \Lxio I_0[ I_0\widetilde{V}_1, \langle V_1\rangle] + \frac{1}{2}\Lxio I_0[I_0\widetilde{V}_1,\widetilde{V}_1]^{\text{osc}} + \frac{1}{2}[\Lxio I_0\widetilde{V}_1, I_0\widetilde{V}_1],V_1\bigg]^{\text{osc}}\nonumber\\
& = \Lxio \bigg(I_0\widetilde{V}_3 +  I_0[ I_0\widetilde{V}_1,\langle V_2\rangle]^{\text{osc}} +I_0[ I_0\widetilde{V}_2,\langle V_1\rangle]^{\text{osc}}+\frac{1}{2} I_0[I_0\widetilde{V}_2,\widetilde{V}_1]^{\text{osc}} + \frac{1}{2}I_0[I_0\widetilde{V}_1,\widetilde{V}_2]^{\text{osc}} \bigg)  \nonumber\\
& +  \frac{1}{2} [\Lxio I_0 \widetilde{V}_1,I_0\widetilde{V}_2]^{\text{osc}}+  \frac{1}{2} [\Lxio I_0 \widetilde{V}_2,I_0\widetilde{V}_1]^{\text{osc}}\nonumber\\
&+ I_0\bigg[ \Lxio I_0[ I_0\widetilde{V}_1, \langle V_1\rangle] + \frac{1}{2}\Lxio I_0[I_0\widetilde{V}_1,\widetilde{V}_1]^{\text{osc}} + \frac{1}{2}[\Lxio I_0\widetilde{V}_1, I_0\widetilde{V}_1],\langle V_1\rangle\bigg]^{\text{osc}}\nonumber\\
&+ I_0\bigg[ \Lxio I_0[ I_0\widetilde{V}_1, \langle V_1\rangle] + \frac{1}{2}\Lxio I_0[I_0\widetilde{V}_1,\widetilde{V}_1]^{\text{osc}} + \frac{1}{2}[\Lxio I_0\widetilde{V}_1, I_0\widetilde{V}_1],\widetilde{V}_1\bigg]^{\text{osc}}\nonumber\\
& = \Lxio \bigg(I_0\widetilde{V}_3 +  I_0[ I_0\widetilde{V}_1,\langle V_2\rangle]^{\text{osc}} +I_0[ I_0\widetilde{V}_2,\langle V_1\rangle]^{\text{osc}}+\frac{1}{2} I_0[I_0\widetilde{V}_2,\widetilde{V}_1]^{\text{osc}} + \frac{1}{2}I_0[I_0\widetilde{V}_1,\widetilde{V}_2]^{\text{osc}} \bigg)  \nonumber\\
&+ \Lxio I_0\bigg[  I_0[ I_0\widetilde{V}_1, \langle V_1\rangle] + \frac{1}{2} I_0[I_0\widetilde{V}_1,\widetilde{V}_1]^{\text{osc}} ,\langle V_1\rangle\bigg]\nonumber\\
& +  \frac{1}{2} [\Lxio I_0 \widetilde{V}_1,I_0\widetilde{V}_2]^{\text{osc}}+  \frac{1}{2} [\Lxio I_0 \widetilde{V}_2,I_0\widetilde{V}_1]^{\text{osc}}+  \frac{1}{2}I_0\bigg[[\Lxio I_0\widetilde{V}_1, I_0\widetilde{V}_1]^{\text{osc}},\langle V_1\rangle\bigg]\nonumber\\
&+ I_0\bigg[ \Lxio I_0[ I_0\widetilde{V}_1, \langle V_1\rangle] + \frac{1}{2}\Lxio I_0[I_0\widetilde{V}_1,\widetilde{V}_1]^{\text{osc}} + \frac{1}{2}[\Lxio I_0\widetilde{V}_1, I_0\widetilde{V}_1],\widetilde{V}_1\bigg]^{\text{osc}}\nonumber\\
& = \Lxio \bigg(I_0\widetilde{V}_3 +  I_0[ I_0\widetilde{V}_1,\langle V_2\rangle]^{\text{osc}} +I_0[ I_0\widetilde{V}_2,\langle V_1\rangle]^{\text{osc}}+\frac{1}{2} I_0[I_0\widetilde{V}_2,\widetilde{V}_1]^{\text{osc}} + \frac{1}{2}I_0[I_0\widetilde{V}_1,\widetilde{V}_2]^{\text{osc}} \bigg)  \nonumber\\
&+ \Lxio I_0\bigg[  I_0[ I_0\widetilde{V}_1, \langle V_1\rangle] + \frac{1}{2} I_0[I_0\widetilde{V}_1,\widetilde{V}_1]^{\text{osc}} ,\langle V_1\rangle\bigg]\nonumber\\
& +  \frac{1}{2} [\Lxio I_0 \widetilde{V}_1,I_0\widetilde{V}_2]^{\text{osc}}+  \frac{1}{2} [\Lxio I_0 \widetilde{V}_2,I_0\widetilde{V}_1]^{\text{osc}}+  \frac{1}{2}I_0\bigg[[\Lxio I_0\widetilde{V}_1, I_0\widetilde{V}_1]^{\text{osc}},\langle V_1\rangle\bigg]\nonumber\\
&+ I_0\bigg[ \Lxio I_0[ I_0\widetilde{V}_1, \langle V_1\rangle] + I_0[\Lxio I_0\widetilde{V}_1,\widetilde{V}_1]^{\text{osc}},\widetilde{V}_1\bigg]^{\text{osc}}\nonumber\\
& = \Lxio \bigg(I_0\widetilde{V}_3 +  I_0[ I_0\widetilde{V}_1,\langle V_2\rangle]^{\text{osc}} +I_0[ I_0\widetilde{V}_2,\langle V_1\rangle]^{\text{osc}}+\frac{1}{2} I_0[I_0\widetilde{V}_2,\widetilde{V}_1]^{\text{osc}} + \frac{1}{2}I_0[I_0\widetilde{V}_1,\widetilde{V}_2]^{\text{osc}} \bigg)  \nonumber\\
&+ \Lxio I_0\bigg[  I_0[ I_0\widetilde{V}_1, \langle V_1\rangle] + \frac{1}{2} I_0[I_0\widetilde{V}_1,\widetilde{V}_1]^{\text{osc}} ,\langle V_1\rangle\bigg]  +\frac{1}{2}\Lxio I_0[I_0\widetilde{V}_1,[I_0\widetilde{V}_1,\langle V_1\rangle]]^{\text{osc}}\nonumber\\
& +  \frac{1}{2} [\Lxio I_0 \widetilde{V}_1,I_0\widetilde{V}_2]^{\text{osc}}+  \frac{1}{2} [\Lxio I_0 \widetilde{V}_2,I_0\widetilde{V}_1]^{\text{osc}}+ [\Lxio I_0[I_0\widetilde{V}_1,\langle V_1\rangle ],I_0\widetilde{V}_1]^{\text{osc}}\nonumber\\
&+[I_0[\Lxio I_0\widetilde{V}_1,\widetilde{V}_1]^{\text{osc}},I_0\widetilde{V}_1]^{\text{osc}}+\frac{1}{3}[I_0\widetilde{V}_1,[\Lxio I_0\widetilde{V}_1,I_0\widetilde{V}_1]^{\text{osc}}]^{\text{osc}}\nonumber\\
& -\frac{1}{3}I_0[\langle [\Lxio I_0\widetilde{V}_1,I_0\widetilde{V}_1] \rangle,\widetilde{V}_1] + \frac{1}{3}I_0[\langle [\Lxio I_0\widetilde{V}_1,\widetilde{V}_1]\rangle, I_0\widetilde{V}_1]\nonumber\\
& + \frac{1}{3} \Lxio I_0[I_0\widetilde{V}_1,[I_0\widetilde{V}_1,\widetilde{V}_1]]^{\text{osc}}.
\end{align}

where we have used the identities
\begin{align}
I_0[\Lxio I_0\widetilde{V}_1,V_2]^{\text{osc}}  + I_0[\Lxio I_0\widetilde{V}_2,V_1]^{\text{osc}}& = \Lxio I_0[ I_0\widetilde{V}_1,\langle V_2\rangle]^{\text{osc}}  + \Lxio I_0[ I_0\widetilde{V}_2,\langle V_1\rangle]^{\text{osc}} \nonumber\\
&+  \frac{1}{2}\Lxio I_0[I_0\widetilde{V}_2,\widetilde{V}_1]^{\text{osc}} + \frac{1}{2} \Lxio  I_0[I_0\widetilde{V}_1,\widetilde{V}_2]^{\text{osc}}\nonumber\\
&+  \frac{1}{2} [\Lxio I_0 \widetilde{V}_1,I_0\widetilde{V}_2]^{\text{osc}}+  \frac{1}{2} [\Lxio I_0 \widetilde{V}_2,I_0\widetilde{V}_1]^{\text{osc}},
\end{align}
\begin{align}
\Lxio I_0[I_0\widetilde{V}_1,\widetilde{V}_1]^{\text{osc}} = 2I_0[\Lxio I_0\widetilde{V}_1,\widetilde{V}_1]^{\text{osc}} - [\Lxio I_0\widetilde{V}_1,I_0\widetilde{V}_1]^{\text{osc}},\label{double_id}
\end{align}
\begin{align}
I_0[\Lxio I_0[I_0\widetilde{V}_1,\langle V_1\rangle],\widetilde{V}_1]^{\text{osc}} & = \frac{1}{2}\Lxio I_0[I_0\widetilde{V}_1,[I_0\widetilde{V}_1,\langle V_1\rangle]]^{\text{osc}} + [\Lxio I_0[I_0\widetilde{V}_1,\langle V_1\rangle ],I_0\widetilde{V}_1]^{\text{osc}}\nonumber\\
& - \frac{1}{2}I_0[[\Lxio I_0\widetilde{V}_1,I_0\widetilde{V}_1]^{\text{osc}},\langle V_1\rangle],
\end{align}
and
\begin{align}
I_0[I_0[\Lxio I_0 \widetilde{V}_1,\widetilde{V}_1]^{\text{osc}},\widetilde{V}_1]^{\text{osc}} & = [I_0[\Lxio I_0\widetilde{V}_1,\widetilde{V}_1]^{\text{osc}},I_0\widetilde{V}_1]^{\text{osc}}+\frac{1}{3}[I_0\widetilde{V}_1,[\Lxio I_0\widetilde{V}_1,I_0\widetilde{V}_1]^{\text{osc}}]^{\text{osc}}\nonumber\\
& -\frac{1}{3}I_0[\langle [\Lxio I_0\widetilde{V}_1,I_0\widetilde{V}_1] \rangle,\widetilde{V}_1] + \frac{1}{3}I_0[\langle [\Lxio I_0\widetilde{V}_1,\widetilde{V}_1]\rangle, I_0\widetilde{V}_1]\nonumber\\
& + \frac{1}{3} \Lxio I_0[I_0\widetilde{V}_1,[I_0\widetilde{V}_1,\widetilde{V}_1]]^{\text{osc}}.
\end{align}
Continuing to group similar terms together,
\begin{align}
\widetilde{\xi}_3& = \Lxio \bigg(I_0\widetilde{V}_3 +  I_0[ I_0\widetilde{V}_1,\langle V_2\rangle]^{\text{osc}} +I_0[ I_0\widetilde{V}_2,\langle V_1\rangle]^{\text{osc}}+\frac{1}{2} I_0[I_0\widetilde{V}_2,\widetilde{V}_1]^{\text{osc}} + \frac{1}{2}I_0[I_0\widetilde{V}_1,\widetilde{V}_2]^{\text{osc}} \bigg)  \nonumber\\
&+ \Lxio I_0\bigg[  I_0[ I_0\widetilde{V}_1, \langle V_1\rangle] + \frac{1}{2} I_0[I_0\widetilde{V}_1,\widetilde{V}_1]^{\text{osc}} ,\langle V_1\rangle\bigg]  +\frac{1}{2}\Lxio I_0[I_0\widetilde{V}_1,[I_0\widetilde{V}_1,\langle V_1\rangle]]^{\text{osc}}\nonumber\\
& +  \frac{1}{2} [\Lxio I_0 \widetilde{V}_1,I_0\widetilde{V}_2]^{\text{osc}}+  \frac{1}{2} [\Lxio I_0 \widetilde{V}_2,I_0\widetilde{V}_1]^{\text{osc}}+ [\Lxio I_0[I_0\widetilde{V}_1,\langle V_1\rangle ],I_0\widetilde{V}_1]^{\text{osc}}\nonumber\\
&+[I_0[\Lxio I_0\widetilde{V}_1,\widetilde{V}_1]^{\text{osc}},I_0\widetilde{V}_1]^{\text{osc}}+\frac{1}{3}[I_0\widetilde{V}_1,[\Lxio I_0\widetilde{V}_1,I_0\widetilde{V}_1]^{\text{osc}}]^{\text{osc}}\nonumber\\
& -\frac{1}{3}I_0[\langle [\Lxio I_0\widetilde{V}_1,I_0\widetilde{V}_1] \rangle,\widetilde{V}_1] + \frac{1}{3}I_0[\langle [\Lxio I_0\widetilde{V}_1,\widetilde{V}_1]\rangle, I_0\widetilde{V}_1]\nonumber\\
& + \frac{1}{3} \Lxio I_0[I_0\widetilde{V}_1,[I_0\widetilde{V}_1,\widetilde{V}_1]]^{\text{osc}}\nonumber\\
& = \Lxio\bigg(I_0\widetilde{V}_3 +  I_0[ I_0\widetilde{V}_1,\langle V_2\rangle]^{\text{osc}} +I_0[ I_0\widetilde{V}_2,\langle V_1\rangle]^{\text{osc}}\nonumber\\
&\quad \hphantom{\Lxio}+\frac{1}{2} I_0[I_0\widetilde{V}_2,\widetilde{V}_1]^{\text{osc}}+ \frac{1}{2}I_0[I_0\widetilde{V}_1,\widetilde{V}_2]^{\text{osc}}+\frac{1}{3}I_0[I_0\widetilde{V}_1,[I_0\widetilde{V}_1,\widetilde{V}_1]]^{\text{osc}}\nonumber\\
&\quad\hphantom{\Lxio}+I_0[  I_0[ I_0\widetilde{V}_1, \langle V_1\rangle] ,\langle V_1\rangle]
+ \frac{1}{2}I_0[    I_0[I_0\widetilde{V}_1,\widetilde{V}_1]^{\text{osc}} ,\langle V_1\rangle]\nonumber\\
 &\quad\hphantom{\Lxio}+\frac{1}{2}I_0[I_0\widetilde{V}_1,[I_0\widetilde{V}_1,\langle V_1\rangle]]^{\text{osc}}\bigg)\nonumber\\
 & +  \frac{1}{2} [\Lxio I_0 \widetilde{V}_1,I_0\widetilde{V}_2]^{\text{osc}}+  \frac{1}{2} [\Lxio I_0 \widetilde{V}_2,I_0\widetilde{V}_1]^{\text{osc}}+ [\Lxio I_0[I_0\widetilde{V}_1,\langle V_1\rangle ],I_0\widetilde{V}_1]^{\text{osc}}\nonumber\\
 &+[I_0[\Lxio I_0\widetilde{V}_1,\widetilde{V}_1]^{\text{osc}},I_0\widetilde{V}_1]^{\text{osc}}+I_0[\langle [\Lxio I_0\widetilde{V}_1,\widetilde{V}_1]\rangle, I_0\widetilde{V}_1]\nonumber\\
 &+\frac{1}{3}[I_0\widetilde{V}_1,[\Lxio I_0\widetilde{V}_1,I_0\widetilde{V}_1]]^{\text{osc}}.
\end{align}
Combining this expression with Eq.\,\eqref{xi3_mean} for $\langle \xi_3\rangle$ and again using the identity \eqref{double_id} gives Eq.\,\eqref{xi3_formula}, as desired.



\bibliographystyle{jpp}

\bibliography{cumulative_bib_file}

\end{document}